\documentclass[a4paper, UKenglish, cleveref, autoref, thm-restate, pdfa]{lipics-v2021}

\pdfoutput=1 % for arXiv
\hideLIPIcs % remove references to LIPIcs series
\nolinenumbers % remove line numbers

\usepackage{csquotes}
\usepackage{mathtools}
\usepackage{stmaryrd}
\usepackage{tikz}
\usepackage{todonotes}
\usepackage{pifont}
\usepackage{hhline}
\usepackage{colortbl}
\usepackage{booktabs}
\usepackage{hyperref}

\usetikzlibrary{arrows, arrows.meta}
\usetikzlibrary{backgrounds}
\usetikzlibrary{calc}
\usetikzlibrary{decorations.markings, decorations.pathmorphing, decorations.pathreplacing}
\usetikzlibrary{fit}
\usetikzlibrary{matrix}
\usetikzlibrary{patterns}
\usetikzlibrary{positioning}
\usetikzlibrary{shapes}
\usetikzlibrary{tikzmark}

\theoremstyle{definition}

\definecolor{efgames-red}{RGB}{255, 0, 0}
\definecolor{efgames-blue}{RGB}{0, 0, 255}
\definecolor{efgames-gray}{RGB}{195, 195, 195}

\definecolor{c1}{rgb}{0.36, 0.54, 0.66}
\newcommand{\cce}{\cellcolor{c1!60}}
\colorlet{c1'}{c1!60}
\definecolor{c2}{rgb}{0.7, 0.7, 0.7}
\newcommand{\ccz}{\cellcolor{c2}}
\definecolor{c3}{rgb}{0.75, 0.58, 0.89}
\newcommand{\ccd}{\cellcolor{c3!50}}
\colorlet{c3'}{c3!50}
\definecolor{c4}{rgb}{0.0, 0.26, 0.15}
\newcommand{\ccv}{\cellcolor{c4!40}}
\colorlet{c4'}{c4!40}
\newcommand{\ccf}{\cellcolor{c1}}
\colorlet{c5'}{c1}

\tikzset{
	graph edge/.style = {draw, thick},
}

\newcommand{\Lra}{\Longrightarrow}

\newcommand{\Iff}{\Longleftrightarrow}

\newcommand{\Trop}{\mathbb{T}}
\newcommand{\Vit}{\mathbb{V}}

\newcommand{\Ninf}{\mathbb{N}^\infty}
\newcommand{\Bool}{\mathbb{B}}

\newcommand{\Lukas}{\mathbb{L}}
\newcommand{\Doubt}{\mathbb{D}}

\newcommand{\WW}{\mathbb{W}}
\newcommand{\Trio}{\textrm{Trio}}

\newcommand{\PosBool}{\textrm{PosBool}}
\newcommand{\Sorb}{\mathbb{S}}
\newcommand{\Sinf}{\mathbb{S}^\infty}

\newcommand{\calS}{\mathcal{S}}
\newcommand{\Semiring}{\ensuremath{\calS}}
\newcommand{\Semi}{\Semiring}

\renewcommand{\phi}{\varphi}
\renewcommand{\theta}{\vartheta}
\renewcommand{\epsilon}{\varepsilon}

\newcommand*{\N}{\ensuremath{\mathbb{N}}}

\newcommand*{\V}{\ensuremath{\mathbb{V}}}

\newcommand*{\R}{\ensuremath{\mathbb{R}}}
\newcommand*{\WX}{\ensuremath{\mathbb{W}[X]}}
\newcommand*{\llb}{\ensuremath{\llbracket}}
\newcommand*{\rrb}{\ensuremath{\rrbracket}}
\newcommand{\cmark}{\ding{51}}
\newcommand{\xmark}{\ding{55}}
\renewcommand{\AA}{{\mathfrak A}}
\newcommand{\BB}{{\mathfrak B}}
\newcommand{\HG}{\ensuremath{\emph{HG}}}

\DeclareMathOperator{\fo}{FO}
\DeclareMathOperator{\qr}{qr}
\DeclareMathOperator{\lit}{Lit}
\DeclareMathOperator{\img}{\mathsf{img}}

\newcommand*{\referto}[2]{\hyperref[#1]{#2~\ref*{#1}}}

\title{Ehrenfeucht--Fraïssé Games in Semiring Semantics}

\author{Sophie Brinke}{%
	RWTH Aachen University, Germany%
}{%
	brinke@logic.rwth-aachen.de%
}{}{}

\author{Erich Grädel}{%
	RWTH Aachen University, Germany%
}{%
	graedel@logic.rwth-aachen.de%
}{%
	https://orcid.org/0000-0002-8950-9991%
}{}

\author{Lovro Mrkonjić}{%
	RWTH Aachen University, Germany%
}{%
	mrkonjic@logic.rwth-aachen.de%
}{%
	https://orcid.org/0000-0001-8812-7185%
}{}

\authorrunning{S. Brinke, E. Grädel and L. Mrkonjić}
\Copyright{Sophie Brinke, Erich Grädel, and Lovro Mrkonjić}


\category{}

\relatedversion{}

\begin{document}
	
\maketitle

\begin{abstract}

Ehrenfeucht--Fraïssé games provide a fundamental method for proving elementary equivalence (and equivalence up to
a certain quantifier rank) of relational structures. We investigate the soundness and completeness of this method
in the more general context of semiring semantics. Motivated originally by provenance analysis of database queries, semiring
semantics evaluates logical statements not just by true or false, but by values in some commutative semiring; this can provide
much more detailed information, for instance concerning the combinations of atomic facts that imply the truth of a statement,
or practical information about evaluation costs, confidence scores, access levels or the number of successful evaluation strategies.
There is a wide variety of different semirings that are relevant  for provenance analysis, and the applicability of 
classical logical methods in semiring semantics may strongly depend on the algebraic properties of the underlying semiring. 

While Ehrenfeucht--Fraïssé games are sound and complete for logical equivalences in classical semantics, and thus on the Boolean semiring,
this is in general not the case for other semirings. We provide a detailed analysis of the soundness and completeness of model comparison games
on specific semirings, not just for classical Ehrenfeucht--Fraïssé games but also 
for other variants based on bijections or counting.  
For instance, we prove that $m$-move Ehrenfeucht--Fraïssé games are sound
(but in general not complete) for $m$-equivalence on fully idempotent semirings, whereas $m$-move bijection games are sound on all semirings.
Further we show that Ehrenfeucht--Fraïssé games without a fixed restriction on the number of moves are sound for elementary equivalence on a
number of further important semirings, but completeness only holds in rare cases.

Finally we propose a new kind of games, called  \emph{homomorphism games}, which are based on the fact that there exist 
certain rather simple semiring interpretations that are locally very different and can be separated even in a one-move game, but which can be proved to be elementarily equivalent via separating sets of homomorphisms into the Boolean semiring.
We prove that these homomorphism games
provide a sound and complete method for logical equivalences on finite and infinite lattice semirings.
 
\end{abstract}

\section{Introduction}
	
\emph{Semiring provenance} was proposed in 2007 in a seminal paper by Green, Karvounarakis, and Tannen \cite{GreenKarTan07}.
It is based on the idea to annotate the atomic facts in a database by values in some commutative semiring, and to propagate these values through a 
database query, keeping track whether information is used
alternatively (as in disjunctions or existential quantifications) or jointly (as in conjunctions or universal quantifications).
Depending on the chosen semiring, the provenance valuation then gives practical information about a query,
beyond its truth or falsity, for instance concerning the \emph{confidence} that we may have in its truth, the \emph{cost} of its evaluation, 
the number of successful evaluation strategies, and so on. Beyond such provenance evaluations in specific \emph{application semirings}, 
more precise information is obtained by evaluations in \emph{provenance semirings}
of polynomials or formal power series, which permit us to \emph{track} which
atomic facts are used (and how often) to compute the answer to the query.

In databases, semiring provenance has been successfully applied to a number of different scenarios, such as conjunctive queries, positive relational algebra, datalog, nested relations, XML,  SQL-aggregates, graph databases (see, e.g., the surveys \cite{GreenTan17,Glavic21}),
but for a long time, it had essentially been restricted to negation-free query languages.
There have been algebraically interesting attempts to cover difference of relations
\cite{AmsterdamerDeuTan11,GeertsPog10,GeertsUngKarFunChr16,GreenIveTan09} but  
they had not resulted in systematic tracking of \emph{negative information}, and
for quite some time, this has been an obstacle for extending semiring provenance to other branches of logic
in computer science. Indeed, while there are many applications in databases where one can get quite far with
considering only positive information, logical applications in most other areas are based on
formalisms that use negation in an essential way.

A new approach to provenance analysis for languages with negation has been proposed in 2017 by Grädel and Tannen
\cite{GraedelTan17}, based on transformations into negation normal form, quotient semirings of polynomials with dual indeterminates, and a 
close relationship to semiring valuations of games \cite{GraedelTan20}. Since then, semiring provenance has been extended
to a systematic investigation of \emph{semiring semantics} for many logical systems, including first-order logic,
modal logic, description logics, guarded logic and fixed-point logic 
\cite{BourgauxOzaPenPre20, DannertGra19, DannertGra20,DannertGraNaaTan21,GraedelTan17}
and also to a general method for strategy analysis in games \cite{GraedelTan20,GraedelLueNaa21}. 

In classical semantics, a model $\AA$ of a formula $\phi$
assigns to each (instantiated) literal a Boolean value.
$\Semi$-interpretations $\pi$, for a suitable semiring $\Semi$, generalise this
by assigning to each such literal a semiring value from $\Semi$.
We interpret $0$ as \emph{false} and all other semiring values as \emph{nuances of true}, or more accurately,
\emph{true, with additional information}. In this context, classical semantics corresponds to semiring semantics on
the Boolean semiring $\Bool = (\{0,1\}, \lor, \land, 0, 1)$, the Viterbi semiring $\Vit = ([0,1], \max, \cdot, 0, 1)$ can model \emph{confidence} scores, 
the tropical semiring $\Trop= (\mathbb{R}_{+}^{\infty},\min,+,\infty,0)$
is used for cost analysis, and min-max-semirings $(A, \max, \min, a, b)$ for a totally ordered set $(A,<)$ can model different access levels.
Other interesting semirings are the Łukasiewicz semiring~$\Lukas$,
used in many-valued logic, and its dual~$\Doubt$,
which we call the semiring of doubt.
Provenance semirings of polynomials, such as $\N[X]$, \emph{track} certain literals by mapping them to different indeterminates. The overall value of  a formula is then a polynomial that describes precisely what combinations of literals imply  the truth of the formula. There are other provenance semirings, obtained from $\N[X]$ by dropping coefficients and/or exponents or by absorption, to get semirings $\Bool[X],\Trio[X], \WW[X],\Sorb[X]$ and
$\PosBool[X]$.
Algebraically, these are quotient semirings obtained from $\N[X]$ by factorisation via suitable congruences. 
They are less informative than~$\N[X]$ (which is the free semiring generated by $X$), but have specific algebraic properties and admit simpler evaluation procedures. For applications to infinite universes, and for stronger logics than first-order logic, provenance semirings with more general objects than polynomials are needed, such as $\Ninf[\![X]\!]$, the semirings of formal power series, and $\Sinf[X|$, the semirings of generalised absorptive polynomials
with potentially infinite exponents, which are fundamental for semiring semantics of fixed-point logics \cite{GraedelTan20,DannertGraNaaTan21}.

\medskip	The development of semiring semantics raises the question to what extent classical techniques and results of logic extend to semiring semantics, 
and how this depends on the algebraic properties of the underlying semiring, and this paper is part of a general research programme that
explores such questions. In previous investigations, we have studied, for instance,  the relationship between elementary equivalence and isomorphism
for finite semiring interpretations and their definability up to isomorphism \cite{GraedelMrk21}, 0-1 laws \cite{GraedelHelNaaWil22},
and locality properties as given by the theorems of Gaifman and Hanf \cite{BiziereGraNaa23}.
In all these studies, it has turned out that classical methods of mathematical logic can be extended to semiring semantics for certain semirings,
but that they fail for others. Further, these questions are often surprisingly difficult: even quite simple facts of logic in the standard Boolean semantics
become interesting research problems for semirings, and they often require completely new methods.

\medskip	The objective of this paper is to study the applicability of Ehrenfeucht--Fraïssé games --- and related model comparison games --- as a method for proving
elementary equivalence (i.e. indistinguishability by first-order sentences, denoted $\equiv$) and $m$-equivalence (i.e.  indistinguishability by sentences
of quantifier rank up to $m$, denoted $\equiv_m$) in semiring semantics. Let us recall the classical Ehrenfeucht--Fraïssé Theorem\footnote{Detailed definitions of all notions will be given in \cref{sect:defs}.} (see e.g. \cite{EbbinghausFlu99}).
	
\begin{theorem}[Ehrenfeucht--Fraïssé]
Let $\tau$ be a finite relational vocabulary.
For any two $\tau$-structures $\AA$ and $\BB$,
and for all $m \in \N$, the following statements are equivalent:

\begin{bracketenumerate}
\item
    ${\mathfrak A} \equiv_m {\mathfrak B}$;

\item
    Player II (Duplicator) has a winning strategy for the game $G_m( \AA,\BB)$;

\item
    There exists an $m$-back-and-forth system $(I_j)_{j \leq m}$ for $\AA$ and $\BB$;

\item
    $\BB\models\chi^m_\AA$, where $\chi^m_\AA$ is
    the characteristic sentence of quantifier rank $m$ for $\AA$.
\end{bracketenumerate}
\end{theorem}
	
In semiring semantics, the structures $\AA$ and $\BB$ are generalised to
(model-defining) semiring interpretations~$\pi_A$ and~$\pi_B$
mapping instantiated $\tau$-literals into a semiring $\Semi$.
The notions of $m$-equivalence, local isomorphisms, Ehrenfeucht--Fraïssé games, and back-and-forth systems all generalise in a straightforward way to $\Semi$-interpretations,
for any semiring $\Semi$ (see \cref{sect:defs}). Also the observation that $m$-back-and-forth systems can be viewed as algebraic descriptions of winning strategies of Player II in $m$-turn Ehrenfeucht--Fraïssé games holds for arbitrary semiring interpretations, i.e.\@ the equivalence (2) $\Leftrightarrow$ (3) holds for any semiring.
The notion of characteristic sentences will be discussed later (e.g.\@ in \cref{sect:homgame} and the proof of \cref{thm-soundcomplNNX}). Our main concern is the relationship between (1) and (2), or equivalently (1) and (3).
We shall have to consider both directions separately.
	
\begin{definition}
Let $\Semi$ be an arbitrary commutative semiring. We say that

\begin{bracketenumerate}
\item
    \emph{$G_m$ is sound for $\equiv_m$ on $\Semi$}
    if for any pair $\pi_A, \pi_B$ of model-defining $\Semi$-interpretations,
    the existence of a winning strategy of Player~II for $G_m(\pi_A,\pi_B)$
    implies that $\pi_A \equiv_m \pi_B$;

\item
    \emph{$G_m$ is complete for $\equiv_m$ on $\Semi$}
    if for any pair $\pi_A,\pi_B$ of model-defining $\Semi$-interpretations
    such that $\pi_A\equiv_m\pi_B$, Player~II has a winning strategy for $G_m(\pi_A,\pi_B)$. 
\end{bracketenumerate}
\end{definition}
	
In this terminology, the Ehrenfeucht--Fraïssé Theorem says that for every $m$, $G_m$ is both sound and complete for $\equiv_m$ on the Boolean semiring.
However, we shall prove that the Boolean semiring is the only semiring with this property, and for general semirings, the games $G_m$ need be neither
sound nor complete. But there are also positive results, and the detailed study of soundness and completeness of Ehrenfeucht--Fraïssé games
on semirings is quite interesting and diverse. For instance, we shall prove that $G_m$ is sound for $\equiv_m$  precisely on
\emph{fully idempotent} semirings (where both semiring operations are idempotent, i.e. $a+a=a\cdot a=a$ for all $a$). 
Examples of fully idempotent semirings include all min-max semirings, the more general lattice semirings, and the semirings $\operatorname{PosBool}[X]$ of irredundant positive Boolean DNF-formulae. Conversely, already the game $G_1$ is unsound for $\equiv_1$ on all semirings that are not fully idempotent.
On the other side, the games $G_m$ are complete for~$\equiv_m$ on the natural semiring $\N$.
We shall then turn to more powerful games, which are more difficult to win for Duplicator, but if she wins, stronger results follow. 
In particular, we study the general Ehrenfeucht--Fraïssé game $G(\pi_A,\pi_B)$ where Spoiler
can choose a number~$m$, and then the game $G_m(\pi_A,\pi_B)$ is played. If, on a semiring $\Semi$,  $G_m$ sound for $\equiv_m$ for all $m$, then
\begin{align*} 
	\text{II wins } G(\pi_A,\pi_B)\ &\Iff\  \text{II wins } G_m(\pi_A,\pi_B)\text{ for all }m \ \Lra\ \pi_A\equiv_m\pi_B \text{ for all }m \\ 
	&\Iff\  \pi_A\equiv\pi_B.\end{align*}
Thus, the soundness of all games $G_m$ implies the soundness of $G$.
The converse is not true; there are semirings on
which $G$ is sound for $\equiv$, although the games $G_m$ are unsound for $\equiv_m$. 
Trivially, this is the case for semirings that
do not admit interpretations with infinite universes
due to the impossibility of infinite sums or products,
such as $\N$ or the provenance semirings $\Bool[X], \mathbb{S}[X]$ and $\N[X]$. 
Clearly, on finite semiring interpretations, Player II can win the
game $G(\pi_A,\pi_B)$ only if $\pi_A$ and $\pi_B$ are isomorphic, and hence also elementarily equivalent.
More interesting cases include semirings that are not idempotent, but
$n$-idempotent for some  $n$, which means that adding or multiplying any element repeatedly with itself stabilises after at most $n$ steps.
An example of different nature is the semiring $\N^\infty$ obtained by extending the natural semiring $\N$ with an infinite element $\infty$.
But there also exist a number of semirings on which the unrestricted Ehrenfeucht--Fraïssé game $G$ is unsound for elementary equivalence,
including the semirings $\Vit,\Trop,\Lukas$ and $\Doubt$.
Further we shall consider \emph{bijection} and \emph{counting games}, which are simplified variants of the pebble games invented by Hella \cite{Hella92}
and by Immerman and Lander \cite{ImmermanLan90} for  bounded-variable logics with counting. 
Actually the $m$-move bijection games $BG_m$
and counting games $CG_m$ are equivalent, and they turn out to be sound for~$\equiv_m$ on \emph{every} semiring. However, with few exceptions, such as the semirings~$\N$ and~$\N[X],$ they are not complete.
A reason to study counting games separately is that they can be parametrised to games $CG^n_m$, in which the sets chosen by the players
during a move can have at most $n$ elements. For $n=1$ this is the same as the standard  Ehrenfeucht--Fraïssé $G_m$, so in general,
the games $CG^n_m$ are between~$G_m$ and $BG_m$, concerning the difficulty of Player II to win. It turns out that the games~$CG^n_m$
are sound for~$\equiv_m$ on $n$-idempotent semirings.

It remains to study the completeness and incompleteness of $G_m$ for $\equiv_m$. 
On many semirings $\Semi$, the methods established in 
\cite{GraedelMrk21} permit us to construct elementarily equivalent $\Semi$-interpretations $\pi_A\equiv \pi_B$, although locally
some elements of $\pi_A$ look different from all elements of $\pi_B$, so that
Spoiler wins $G_m(\pi_A,\pi_B)$ for some (small) $m$, in fact often for $m=1$.
The game $G_m$ is then incomplete for $\equiv_m$, and the game $G$ is incomplete for $\equiv$.
Since the games $CG^n_m$ and $BG_m$ are more difficult to win for Player II than $G_m$,
they are incomplete as well. This approach successfully works for the semirings $\Vit,\Trop,\Lukas,\Doubt,\N^\infty,\mathbb{W}[X], \mathbb{S}[X], \Bool [X]$, and $\mathbb{S}^\infty[X]$.
In some cases the interpretations $\pi_A,\pi_B$ must be infinite.
Soundness and completeness results of these games
are summarised in \cref{sum-Gm-BGm}.

\begin{figure}[ht]
\renewcommand{\arraystretch}{1.2}
\begin{tabular}{>{\raggedleft\arraybackslash}m{0.9cm}>{\centering\arraybackslash}m{2.2cm}>{\centering\arraybackslash}m{1.7cm}>{\centering\arraybackslash}m{1.5cm}>{\centering\arraybackslash}m{1.5cm}>{\centering\arraybackslash}m{1.4cm}>{\centering\arraybackslash}m{1.4cm}}
\multicolumn{2}{m{3.5cm}}{Application semirings:} &
\footnotesize $\Semi \not\cong \Bool$ fully idempotent &
\small $\Trop \cong \Vit$ &
\small $\mathbb{L} \cong \mathbb{D}$ &
\small $\N$ &  \small $\N^\infty$ \\
\addlinespace[2pt]\hhline{-------}\addlinespace[2pt]
\multirow{4}{*}{\rotatebox[origin=c]{90}{\small\strut Soundness}} &
\small $G_m$ for $\equiv_m$ &
\cce \hyperref[thm-soundIdem]{~\hfill\cmark\hfill~} &
\cce \hyperref[thm-soundIdem]{~\hfill\xmark\hfill~} &
\cce \hyperref[thm-soundIdem]{~\hfill\xmark\hfill~} &
\cce \hyperref[thm-soundIdem]{~\hfill\xmark\hfill~} &
\cce \hyperref[thm-soundIdem]{~\hfill\xmark\hfill~} \\
&
\small $CG^n_m$ for $\equiv_m$ &
\ccf \hyperref[thm-soundCGm]{~\hfill\cmark\hfill~} &
\ccf \hyperref[thm-soundCGm]{~\hfill\xmark\hfill~} &
\ccf \hyperref[thm-soundCGm]{~\hfill\xmark\hfill~} &
\ccf \hyperref[thm-soundCGm]{~\hfill\xmark\hfill~} &
\ccf \hyperref[thm-soundCGm]{~\hfill\xmark\hfill~} \\
&
\small $BG_m$ for $\equiv_m$ &
\ccd \hyperref[thm-soundBGm]{~\hfill\cmark\hfill~} &
\ccd \hyperref[thm-soundBGm]{~\hfill\cmark\hfill~} &
\ccd \hyperref[thm-soundBGm]{~\hfill\cmark\hfill~} &
\ccd \hyperref[thm-soundBGm]{~\hfill\cmark\hfill~} &
\ccd \hyperref[thm-soundBGm]{~\hfill\cmark\hfill~} \\
&
\small $G$ for $\equiv$ &
\cce \hyperref[thm-soundIdem]{~\hfill\cmark\hfill~} &
\hyperref[thm-soundVG]{~\hfill\xmark\hfill~} &
\xmark &
\ccz \cmark &
\hyperref[thm-soundNinfG]{~\hfill\cmark\hfill~} \\
\addlinespace[2pt]\hhline{-------}\addlinespace[2pt]
\multirow{4}{*}{\rotatebox[origin=c]{90}{\small\strut Completeness}} &
\small $G_m$ for $\equiv_m$ &
\ccv \hyperref[thm-complIdem]{~\hfill\xmark\hfill~} &
\hyperref[prop-V1equiv]{~\hfill\xmark\hfill~} &
\ccv \xmark &
\hyperref[thm-complN]{~\hfill\cmark\hfill~} &
\hyperref[cor-complNinf]{~\hfill\xmark\hfill~} \\
&
\small $CG^n_m$ for $\equiv_m$ &
\ccv \hyperref[thm-complIdem]{~\hfill\xmark\hfill~} &
\hyperref[prop-V1equiv]{~\hfill\xmark\hfill~} &
\ccv \xmark &
\hyperref[thm-complN]{~\hfill\cmark\hfill~} &
\hyperref[cor-complNinf]{~\hfill\xmark\hfill~} \\
&
\small $BG_m$ for $\equiv_m$ &
\ccv \hyperref[thm-complIdem]{~\hfill\xmark\hfill~} &
\hyperref[prop-V1equiv]{~\hfill\xmark\hfill~} &
\ccv \xmark &
\hyperref[thm-complN]{~\hfill\cmark\hfill~} &
\hyperref[cor-complNinf]{~\hfill\xmark\hfill~} \\
&
\small $G$ for $\equiv$ &
\ccv \hyperref[thm-complIdem]{~\hfill\xmark\hfill~} &
\hyperref[thm-complVG]{~\hfill\xmark\hfill~} &
\ccv \xmark &
\hyperref[thm-complN]{~\hfill\cmark\hfill~} &
\hyperref[cor-complNinf]{~\hfill\xmark\hfill~} \\
\end{tabular}

\bigskip

\medskip

\begin{tabular}{>{\raggedleft\arraybackslash}m{0.9cm}>{\centering\arraybackslash}m{2.2cm}>{\centering\arraybackslash}m{1.7cm}>{\centering\arraybackslash}m{1.5cm}>{\centering\arraybackslash}m{1.5cm}>{\centering\arraybackslash}m{1.4cm}>{\centering\arraybackslash}m{1.4cm}}
\multicolumn{2}{m{3.5cm}}{Provenance semirings:} &
\small $\operatorname{PosBool}[X]$ &
\small $\mathbb{W}[X]$ &
\small $\mathbb{S}[X], \Bool[X]$ &
\small $\N[X]$ &
\small $\Sinf[X]$ \\
\addlinespace[2pt]\hhline{-------}\addlinespace[2pt]
\multirow{4}{*}{\rotatebox[origin=c]{90}{\small\strut Soundness}} &
\small $G_m$ for $\equiv_m$ &
\cce \hyperref[thm-soundIdem]{~\hfill\cmark\hfill~} &
\cce \hyperref[thm-soundIdem]{~\hfill\xmark\hfill~} &
\cce \hyperref[thm-soundIdem]{~\hfill\xmark\hfill~} &
\cce \hyperref[thm-soundIdem]{~\hfill\xmark\hfill~} &
\cce \hyperref[thm-soundIdem]{~\hfill\xmark\hfill~} \\
&
\small $CG^n_m$ for $\equiv_m$ &
\ccf \hyperref[thm-soundCGm]{~\hfill\cmark\hfill~} &
\ccf \hyperref[thm-soundCGm]{~\hfill\cmark\hfill~} &
\ccf \xmark &
\ccf \hyperref[thm-soundCGm]{~\hfill\xmark\hfill~} &
\ccf \hyperref[thm-soundCGm]{~\hfill\xmark\hfill~} \\
&
\small $BG_m$ for $\equiv_m$ &
\ccd \hyperref[thm-soundBGm]{~\hfill\cmark\hfill~} &
\ccd \hyperref[thm-soundBGm]{~\hfill\cmark\hfill~} &
\ccd \hyperref[thm-soundBGm]{~\hfill\cmark\hfill~} &
\ccd \hyperref[thm-soundBGm]{~\hfill\cmark\hfill~} &
\ccd \hyperref[thm-soundBGm]{~\hfill\cmark\hfill~} \\
&
\small $G$ for $\equiv$ &
\cce \hyperref[thm-soundIdem]{~\hfill\cmark\hfill~} &
\ccf \hyperref[thm-soundWXG]{~\hfill\cmark\hfill~} &
\ccz \cmark &
\ccz \cmark &
\hyperref[thm-soundSinfG]{~\hfill\cmark\hfill~} \\
\addlinespace[2pt]\hhline{-------}\addlinespace[2pt]
\multirow{4}{*}{\rotatebox[origin=c]{90}{\small\strut Completeness}} &
\small $G_m$ for $\equiv_m$ &
\ccv \hyperref[thm-complIdem]{~\hfill\xmark\hfill~} &
\ccv \xmark &
\ccv \xmark &
\hyperref[thm-complNX]{~\hfill\cmark\hfill~} &
\ccv \xmark \\
&
\small $CG^n_m$ for $\equiv_m$ &
\ccv \xmark &
\ccv \xmark &
\ccv \xmark &
\hyperref[thm-complNX]{~\hfill\cmark\hfill~} &
\ccv \xmark \\
&
\small $BG_m$ for $\equiv_m$ &
\ccv \hyperref[thm-complIdem]{~\hfill\xmark\hfill~} &
\ccv \xmark &
\ccv \xmark &
\hyperref[thm-complNX]{~\hfill\cmark\hfill~} &
\ccv \xmark \\
&
\small $G$ for $\equiv$ &
\ccv \hyperref[thm-complIdem]{~\hfill\xmark\hfill~} &
\ccv \xmark &
\ccv \xmark &
\hyperref[thm-complNX]{~\hfill\cmark\hfill~} &
\ccv \xmark \\
\end{tabular}

\medskip

\renewcommand{\arraystretch}{1.0}
\caption{
\sethlcolor{c1'} \hl{Due to full idempotence.}
\sethlcolor{c5'} \hl{Due to $n$-idempotence.}
\sethlcolor{c3'} \hl{Holds for any semiring.}
\sethlcolor{c2} \hl{Follows from the finiteness of the universes.}
\sethlcolor{c4'} \hl{Cannot hold since elementary equivalence of finite interpretations does not imply isomorphism.}
}
\label{sum-Gm-BGm}
\end{figure}

The proof that locally different $\Semi$-interpretations are nevertheless elementarily equivalent often proceeds 
via  separating sets of homomorphisms. We use this method to
propose a new kind of games, called  \emph{homomorphism games}, involving the selection of a homomorphism into the Boolean semiring, and
a one-sided winning condition, due to the property that homomorphisms may transfer model-defining $\Semi$-interpretations
into $\Bool$-interpretations that are no longer model-defining.  We prove that these homomorphism games
provide a sound and complete method for proving logical equivalences on finite and infinite lattice semirings.

\section{Semiring semantics}\label{sect:defs}

We briefly summarise semiring semantics for first-order logic, as introduced in \cite{GraedelTan17}, and the resulting generalised notions of isomorphism and equivalence.
	
\begin{definition}[Semiring]
	A commutative semiring is an algebraic structure $\Semi = (S, +, \cdot, 0, 1)$ with $0  \neq 1$, such that $(S, +, 0)$ and $(S, \cdot, 1)$ are commutative monoids, $\cdot$ distributes over $+$, and $0 \cdot s = s \cdot 0 = 0$.
\end{definition}

A commutative semiring is  \emph{naturally ordered} (by addition) if $s \leq t :\Leftrightarrow \exists r (s+r = t)$ defines a partial order. In particular, this excludes rings.
We only consider commutative and naturally ordered semirings and simply refer to them as \emph{semirings}.  
A semiring $\Semi$ is \emph{idempotent} if $s+s=s$ for each $s \in S$ and \emph{multiplicatively idempotent} if $s \cdot s = s$ for all $s \in S$.
If both properties are satisfied, we say that $\Semi$ is fully idempotent.
Finally, $\Semi$ is \emph{absorptive} if $s + s t = s$ for all $s, t \in S$ or, equivalently, if multiplication is decreasing in $\Semi$, i.e. $st \leq s$ for $s, t \in S$. Every absorptive semiring is idempotent.
	
\subparagraph*{Application semirings} There are several applications which can be modelled by semirings and provide useful practical information about the evaluation of a formula. 
\begin{itemize}
	\item A totally ordered set $(S, \leq)$ with least element $s$ and greatest element $t$ induces the \emph{min-max semiring} $(S, \max, \min, s, t)$. It can be used to reason about access levels.
	\item The \emph{tropical semiring} $\Trop = (\R^\infty_+ , \min, +, \infty, 0)$ provides the opportunity to annotate basic facts with a cost which has to be paid for accessing them and realise a cost analysis.
	\item The \emph{Viterbi semiring} $\Vit = ([0,1]_\R,\max,\cdot,0,1)$, which is in fact  isomorphic to $\Trop$ via $y \mapsto -\ln y$
	can be used for  reasoning about confidence. 
	\item An alternative semiring for this is the \emph{Łukasiewicz semiring} $\Lukas= ([0,1]_\R, \max, \odot, 0, 1)$, where multiplication is given by $s \odot t = \max(s+t -1, 0)$. It is isomorphic to the semiring of doubt $\Doubt = ([0,1]_\R, \min, \oplus, 1, 0)$ with $s \oplus t = \min(s+t, 1)$.
	\item The \emph{natural semiring} $\N = (\N, +, \cdot, 0, 1)$ is used to count the number of evaluation strategies proving that a sentence is satisfied. It is also important for bag semantics in databases. 
\end{itemize}
	
\subparagraph*{Provenance semirings}
Provenance semirings of polymomals provide information on which combinations of literals imply the truth of a formula.
The universal provenance semiring is the semiring $\N[X]$ of multivariate polynomials with indeterminates from $X$ and coefficients from $\N$. Other provenance 
semirings are obtained as quotient semirings of $\N[X]$ induced by congruences for (full) idempotence and absorption. The resulting provenance values are less informative but their computation is more efficient. 
\begin{itemize}
	\item By dropping coefficients from $\N[X]$, we get the free idempotent semiring $\Bool[X]$ whose elements are finite sets of monomials. It is the quotient induced by $x + x \thicksim x$.
	\item If, in addition, exponents are dropped, we obtain the Why-semiring $\WX$ of finite sums of monomials that are linear in each argument.
	\item The free absorptive semiring $\Sorb[X]$ consists of $0,1$ and all antichains of monomials with respect to the absorption order $\succcurlyeq$. A monomial $m_1$ absorbs $m_2$, denoted $m_1 \succcurlyeq m_2$, if it has smaller exponents, i.e. $m_2 =m \cdot m_1$ for some monomial $m$.
	\item Finally, the lattice semiring $\PosBool[X]$ freely generated by the set $X$ arises from $\Sorb[X]$ by collapsing exponents.
\end{itemize}

For a given finite relational vocabulary $\tau$, we denote by $\lit_n (\tau)$ the set of literals $R \bar{x}$ and $\neg R \bar{x}$ where $R \in \tau$ and $\bar{x}$ is a tuple of variables from $\{x_1, \dots , x_n\}$. The set $\lit_A (\tau)$ refers to literals $R \bar{a}$ and $\lnot R \bar{a}$ that are instantiated with elements from a universe $A$.

\begin{definition}[$\Semi$-interpretation]
	Given a semiring $\Semi$, a mapping $\pi \colon \lit_A (\tau) \to \Semi$ is an $\Semi$-interpretation (of vocabulary $\tau$ and universe $A$). We say that $\Semi$ is model-defining if exactly one of the values $\pi(L)$ and $\pi(\overline{L})$ is zero for any pair of complementary literals $L, \overline{L} \in \lit_A(\tau)$.
\end{definition}

An $\Semi$-interpretation $\pi \colon \lit_A (\tau) \to \Semi$ inductively extends to valuations $\pi \llb \varphi (\bar{a}) \rrb$ of instantiated first-order formulae $\varphi (\bar{x})$ in negation normal form. Equalities are interpreted by their truth value, that is $\pi \llb a = a \rrb := 1 $ and $\pi \llb a = b \rrb := 0$ for $a \neq b$ (and analogously for inequalities). Based on that, the semantics of disjunction and existential quantifiers is defined via sums, while conjunctions and universal quantifiers are interpreted as products.
\begin{alignat*}{3}
	\pi \llb \psi (\bar{a}) \vee \theta (\bar{a}) \rrb &:= \pi \llb \psi (\bar{a}) \rrb +  \pi \llb  \theta (\bar{a}) \rrb \hspace*{1cm} &\pi \llb \psi (\bar{a}) \wedge \theta (\bar{a}) \rrb &:= \pi \llb \psi (\bar{a}) \rrb \cdot  \pi \llb  \theta (\bar{a}) \rrb \\
	\pi \llb \exists x \psi (\bar{a},x) \rrb &:= \sum_{a \in A} \pi \llb \psi (\bar{a}, a) \rrb \hspace*{1cm}  &  \pi \llb \forall x \psi (\bar{a},x) \rrb &:= \prod_{a \in A} \pi \llb \psi (\bar{a}, a) \rrb
\end{alignat*}

\begin{lemma}[Fundamental Property]
	Let $\pi \colon \lit_A(\tau) \to \Semi$ be an $\Semi$-interpretation and $h \colon \Semi \to \mathcal{T}$ be a semiring homomorphism. Then, $(h \circ \pi)$ is a $\mathcal{T}$-interpretation and it holds that
	$h(\pi \llb \phi (\bar{a}) \rrb ) = (h \circ \pi) \llb \phi (\bar{a}) \rrb$ for all $\phi ( \bar{x}) \in \fo (\tau)$ and instantiations $\bar{a} \subseteq A$.
\end{lemma}

Basic model theoretic concepts such as equivalence and isomorphism naturally generalise to semiring semantics and yield more fine-grained notions. Given a mapping $\sigma \colon A \to B$ and some $L \in \lit_A (\tau)$, we denote by $\sigma (L)$ the $\tau$-literal over $B$ which arises from $L$ by replacing each occurrence of $a \in A$ with $\sigma (a) \in B$.

\begin{definition}[Isomorphism]
	$\Semi$-interpretations $\pi_A \colon \lit_A(\tau) \to \Semi$ and $\pi_B \colon \lit_B(\tau) \to \Semi$ are isomorphic, denoted as $\pi_A \cong \pi_B$, if there is a bijection $\sigma \colon A \to B$ such that $\pi_A(L) = \pi_B(\sigma(L))$ for all $L \in \lit_A(\tau)$. A mapping $\sigma \colon \bar{a} \mapsto \bar{b}$  is a local isomorphism between $\pi_A$ and $\pi_B$ if it is an isomorphism between the subinterpretations $\pi_A \vert_{\lit_{\bar{a}} (\tau)}$ and $\pi_B \vert_{\lit_{\bar{b}} (\tau)}$.
\end{definition}

\begin{definition}[Elementary equivalence]
	Two $\Semi$-interpretations $\pi_A \colon \lit_A (\tau) \to \Semi$ and $\pi_B \colon \lit_B (\tau) \to \Semi$ with elements $\bar{a} \in A^n$ and $\bar{b} \in B^n$ are elementarily equivalent, denoted $(\pi_A, \bar{a}) \equiv (\pi_B, \bar{b})$, if $\pi_A \llb \phi (\bar{a}) \rrb = \pi_B \llb \phi (\bar{b}) \rrb$  for all $\phi(\bar{x}) \in \fo(\tau)$. They are $m$-equivalent, denoted $(\pi_A,\bar{a}) \equiv_m (\pi_B,\bar{b})$, if the above holds for all $\phi(\bar{x})$ with quantifier rank at most $m$.
\end{definition}

As in classical semantics, isomorphic $\Semi$-interpretations are elementarily equivalent. The converse, however, marks an important difference to Boolean semantics;  it fails for a number of semirings, including all min-max semirings with at least three elements, while it still holds on other semirings such as $\Trop, \Vit, \N$ and $\N[X]$ (see \cite{GraedelMrk21}).

\section{\texorpdfstring{$m$}{m}-turn Ehrenfeucht--Fraïssé games} \label{sec:mTurnEF}

Given that the notion of local isomorphisms  extends in a straightforward way from structures to semiring interpretations, we also obtain
Ehrenfeucht--Fraïssé games $G_m (\pi_A, \pi_B)$ played on $\Semi$-interpretations $\pi_A$, $\pi_B$:
In the $i$-th turn, Spoiler chooses some element $a_i \in A$ or $b_i \in B$, and Duplicator answers with an element in the other $\Semi$-interpretation;
the play then continues with the subgame $G_{m-i} (\pi_A, a_1, \dots, a_i, \pi_B, b_1, \dots, b_i)$.
After $m$ moves, tuples $\bar a=(a_1,\dots,a_m)$ in $A$ and $\bar b=(b_1,\dots,b_m)$ in $B$ have been selected,
and Duplicator wins the play if $\sigma \colon \bar{a} \mapsto \bar{b}$ is a local isomorphism. 
	
However, while classical structures $\AA$ and $\BB$ are separated by a formula $\exists x \psi(x)$ or $\forall x \psi (x)$ if, and only if, there is some $a \in A$ (or $b \in B$) such that for all $b \in B$ (or $a \in A$, respectively) the formula $\psi(x)$ separates $(\AA, a)$ from $(\BB,b)$, neither of the implications translates to semiring semantics. Hence, in contrast to the game $G_m$ played on classical structures, already very simple semiring interpretations\footnote{We describe semiring interpretations over a monadic vocabulary by tables, whose rows are indexed by elements of the universe, and columns by the 
predicate symbols and their negations, such that the entry for row $a$ and column $P$ has the semiring value of the literal $Pa$.}
illustrate that  $G_m (\pi_A, \pi_B)$ is in general neither sound nor complete for $\equiv_m$.
	
\medskip
\begin{minipage}[t]{0.46\textwidth}
	\begin{center}
		{\small
			$(\N, +, \cdot, 0, 1)$ \\
			\ \\
			$\pi_A:$
			\begin{tabular}{c||c|c}
				$A$ & $R$ & $\lnot R$  \\
				\hline
				\hline
				$a_1$ & $1$ & $0$  \\
				\hline
				$a_2$ & $1$ & $0$ \\
				\hline
				$a_3$ & $2$ & $0$ \\
			\end{tabular}
			\hspace*{1mm}
			$\pi_B:$
			\begin{tabular}{c||c|c}
				$B$ & $R$ & $\lnot R$  \\
				\hline
				\hline
				$b_1$ & $1$ & $0$  \\
				\hline
				$b_2$ & $2$ & $0$  \\
				\hline
				$b_3$ & $2$ & $0$  \\
			\end{tabular} 
			\ \\ \ \\
			$\pi_A\llb \exists x Rx \rrb = 4 \neq 5 = \pi_B \llb \exists x Rx \rrb$
		}
	\end{center}
\end{minipage}
\hfill\vline\hfill
\begin{minipage}[t]{0.46\textwidth}
	\begin{center}
		{\small
			$(\{0,1,2,3,4\}, \max, \min, 0, 4)$ \\
			\ \\
			$\pi_A: $
			\begin{tabular}{c||c|c}
				$A$ & $R$ & $\lnot R$  \\
				\hline
				\hline
				$a_1$ & $1$ & $0$  \\
				\hline
				$a_2$ & $2$ & $0$ \\
				\hline
				$a_3$ & $4$ & $0$ \\
			\end{tabular}
			\hspace*{1mm}
			$\pi_B:$
			\begin{tabular}{c||c|c}
				$B$ & $R$ & $\lnot R$  \\
				\hline
				\hline
				$b_1$ & $1$ & $0$  \\
				\hline
				$b_2$ & $3$ & $0$  \\
				\hline
				$b_3$ & $4$ & $0$  \\
			\end{tabular} 
			\ \\ \ \\
			$\pi_A\llb \exists x Rx \rrb =  4 = \pi_B \llb \exists x Rx \rrb$ \\
			$\pi_A\llb \forall x Rx \rrb =  1 = \pi_B \llb \forall x Rx \rrb$
		}
	\end{center}
\end{minipage}

\medskip
This suggests that the direct adaptation of the game rules poses problems and raises the question on which semirings the game $G_m$ is sound, and on which it is complete for $\equiv_m$. In particular, we aim to relate this to the algebraic properties of the underlying semiring.
	
\subsection{Soundness of the games and counting in semirings}
	
The fact that quantifiers in classical semantics do not capture counting is one of the central limitations of the expressive power of  first-order logic.
However, in semiring semantics, this is more complicated: Given a formula $\psi(x)$ and some $s \in \Semi$, the number of $a \in A$ such that $\pi \llb \psi (a) \rrb =s$ may affect both $\pi \llb \exists x \psi (x) \rrb$ and $\pi \llb \forall x \psi (x) \rrb$.
Only in fully idempotent semirings unequal sums or products can be attributed to differing sets of summands or factors, which causes full idempotence to be 
a necessary and sufficient condition for the soundness of $G_m$.

\begin{theorem} \label{thm-soundIdem}
The games $G_m$ are sound for $\equiv_m$ on a semiring $\Semi$ and all $m \in \N$ if, and only if, $\Semi$ is fully idempotent.
\end{theorem}
	
\begin{proof}
	$(\Leftarrow)$: Suppose that $\Semi$ is fully idempotent. 
	Based on a separating formula $\varphi (\bar{x}) \in \fo(\tau)$ with $\pi_A \llb \varphi (\bar{a}) \rrb \neq \pi_B \llb \varphi (\bar{b}) \rrb$ and $\qr (\varphi (\bar{x})) \leq m$ where $\bar{a} \in A^n$ and $\bar{b} \in B^n$, we construct a winning strategy for Spoiler in the game $G_m(\pi_A, \bar{a}, \pi_B, \bar{b})$ by induction.
	We only consider the cases $\varphi(\bar{x}) = Qx \psi (\bar{x},x)$ with $Q \in \{\exists, \forall\}$ where $\qr(\varphi (\bar{x})) \leq m$.
	It holds that
	\begin{align*}
		\pi_A \llb \exists x \psi(\bar{a},x) \rrb = \sum\limits_{a \in A} \pi_A \llb \psi (\bar{a},a) \rrb &\neq \sum\limits_{b \in B} \pi_B \llb \psi (\bar{b},b) \rrb = \pi_B \llb \exists x \psi(\bar{b},x) \rrb \text{ or} \\
		\pi_A \llb \forall x \psi(\bar{a},x) \rrb = \prod\limits_{a \in A} \pi_A \llb \psi (\bar{a},a) \rrb &\neq \prod\limits_{b \in B} \pi_B \llb \psi (\bar{b},b) \rrb = \pi_B \llb \forall x \psi(\bar{b},x) \rrb.
	\end{align*}
	Both cases imply
	$
	\{ \pi_A \llb \psi (\bar{a},a) \rrb \colon a \in A \} \neq \{ \pi_B \llb \psi (\bar{b},b) \rrb \colon b \in B \}
	$
	due to full idempotence.
	Spoiler wins the game $G_m (\pi_A, \bar{a}, \pi_B, \bar{b})$ by choosing some element $a \in A$ or $b \in B$ witnessing this inequality.
	For all possible answers $b \in B$ or $a \in A$, respectively, it holds that $\pi_A \llb \psi (\bar{a}, a) \rrb \neq \pi_B \llb \psi (\bar{b}, b) \rrb$.
	Applying the induction hypothesis yields that Spoiler has a winning strategy for the remaining game $G_{m-1} (\pi_A, \bar{a}, a, \pi_B, \bar{b}, b)$, as $\qr(\psi(\bar{x},x)) \leq m-1$.
	
	$(\Rightarrow)$: If $\Semi$ is not fully idempotent,
	there is some $s \in \Semi$ such that 
	$
	s+s \neq s \; \text{ or } \; s \cdot s \neq s.
	$
	Clearly, Duplicator wins $G_1(\pi_A, \pi_B)$
	on the following $\Semi$-interpretations,
	while $\pi_A \not\equiv_1 \pi_B$ due to
	$
		\pi_A \llb \exists x Rx \rrb = s+s \neq s = \pi_B \llb \exists x Rx \rrb \text{ or }
		\pi_A \llb \forall x Rx \rrb = s \cdot s \neq s = \pi_B \llb \forall x Rx \rrb.
	$
	
	\medskip
		
	\begin{minipage}[c]{0.8\linewidth}
		\centering
		$\pi_A:\quad$
		\begin{tabular}{c || c | c}
			$A$ & $R$ & $\neg R$ \\ \hline\hline
			$a_1$ & s & 0 \\
			$a_2$ & s & 0 \\
		\end{tabular}
		$\quad\quad\quad\pi_B:\quad$
		\begin{tabular}{c || c | c}
			$B$ & $R$ & $\neg R$ \\ \hline\hline
			$b$ & s & 0 \\
		\end{tabular}\hspace*{8mm}
	\end{minipage}
\end{proof}
	
This result motivates the consideration of more powerful games such as $m$-turn bijection games, a variant of the pebble games which, on finite classical structures, characterise $m$-equivalence in $\fo$ with counting quantifiers \cite{Hella92}.

\begin{definition}
		The game $BG_m(\pi_A, \bar{a}, \pi_B, \bar{b})$ uses the same positions and winning condition as $G_m(\pi_A, \bar{a}, \pi_B, \bar{b})$, but in each round Duplicator has to provide a bijection $h \colon A \to B$.
		If such a bijection does not exist, i.e. $|A| \neq |B|$, Spoiler wins immediately.
		Otherwise, Spoiler chooses some $a \in A$ and the pair $(a,h(a))$ is added to the current position.
\end{definition}
	
In contrast to the classical Ehrenfeucht--Fraïssé game, this modification ensures soundness without requiring full idempotence of the underlying semiring.

\begin{theorem} \label{thm-soundBGm}
For every $m\in\N$, the game $BG_m$ is sound for $\equiv_m$ on every semiring $\Semi$.
\end{theorem}
	
\begin{proof}
Suppose that $\varphi (\bar{x}) = Q x \psi (\bar{x},x)$ with $Q \in \{\exists, \forall\}$ and $\qr (\varphi(\bar{x})) = m$ separates $(\pi_A, \bar{a})$ from $(\pi_B, \bar{b})$.
For any bijection $h \colon A \to B$ Duplicator may choose in the game $BG_m (\pi_A, \bar{a}, \pi_B, \bar{b})$, there must be some $a_h \in A$ such that $\pi_A \llb \psi (\bar{a},a_h) \rrb \neq \pi_B \llb \psi (\bar{b}, h(a_h)) \rrb$, since otherwise
$\sum\nolimits_{a \in A} \pi_A \llb \psi (\bar{a},a) \rrb = \sum\nolimits_{a \in A} \pi_B \llb \psi (\bar{b},h(a)) \rrb = \sum\nolimits_{b \in B} \pi_B \llb \psi (\bar{b},b) \rrb$ and analogously for products. By choosing $a_h$, Spoiler wins the game by induction.
\end{proof}

While demanding a bijection from Duplicator does ensure the soundness of~$BG_m$,
it is often at the expense of completeness.
This is due to the fact that different multiplicities of a semiring value
in two interpretations do not necessarily imply separability by a first-order sentence.
In particular, this is the case for fully idempotent semirings,
on which the games~$G_m$ are already sound,
but the resulting issues concern other semirings as well.
We illustrate this on the semiring $\mathbb{W}[x, y]$
where the precise numbers of occurrences of single semiring values
may differ in their effect on the separability of the resulting interpretations,
as shown below.

\medskip

\begin{minipage}{0.9\linewidth}
	\centering
	\begin{tabular}{c||c|c}
		$A$ & $R$ & $\lnot R$ \\ \hline\hline
		$a_1$ & $x+y$ & $0$ \\
	\end{tabular}
	\hfill $\not\equiv_1$ \hfill
	\begin{tabular}{c||c|c}
		$B$ & $R$ & $\lnot R$ \\ \hline\hline
		$b_1$ & $x+y$ & $0$ \\ \hline
		$b_2$ & $x+y$ & $0$
	\end{tabular}
	\hfill $\equiv_1$ \hfill
	\begin{tabular}{c||c|c}
		$C$ & $R$ & $\lnot R$ \\ \hline\hline
		$c_1$ & $x+y$ & $0$ \\ \hline
		$c_2$ & $x+y$ & $0$ \\ \hline
		$c_3$ & $x+y$ & $0$
	\end{tabular}
\end{minipage}

\medskip

We observe that the semirings $\mathbb{W}[X]$,
while not being fully idempotent,
for instance due to $(x + y)(x + y) = x + xy + y$,
satisfy a weaker idempotence condition.

\begin{definition}
Let $n \in \N$.
A semiring $\Semi$ is $n$-\emph{idempotent}
if $\sum_{i \in I} s = \sum_{j \in J} s$ and
$\prod_{i \in I} s = \prod_{j \in J} s$ for all $s \in \Semi$
and all index sets $I, J$ such that $\vert I \vert \geq n$ and $\vert J \vert \geq n$.
\end{definition}

It can easily be verified that  $\WX$ is $|X|$-idempotent, as monomials can be seen as sets of variables, and multiplication corresponds to their union.
For such semirings, we want to replace the requirement for Duplicator to provide a bijection  by a weaker requirement that still
maintains soundness. For this, we use counting games, introduced  by Immermann and Lander \cite{ImmermanLan90}, which are equivalent to 
bijection games, but admit a parametrisation by the size of the sets that are picked in each turn.

\begin{definition}
	Let $n \in \N$. In each turn of the game $CG_m^n(\pi_A, \bar{a}, \pi_B, \bar{b})$, Spoiler chooses a set $X \subseteq A$ or $X \subseteq B$  with $|X| \leq n$ and Duplicator has to react with a subset $Y$ of the other universe such that $|X|=|Y|$. Afterwards, Spoiler picks some $y \in Y$, Duplicator must respond with some element $x \in X$ and the pair $(x, y)$, or $(y, x)$, is added to the current position. As before, the winning condition is given by local isomorphism.
\end{definition}

Note that the game $CG^1_m$ corresponds to the classical Ehrenfeucht--Fraïssé game $G_m$ and $1$-idempotence coincides with full idempotence. 
\Cref{thm-soundIdem} can be generalised as follows.
	
\begin{restatable}{theorem}{THMsoundCGm} \label{thm-soundCGm}
	The games $CG_m^n$ are sound for $\equiv_m$ exactly on $n$-idempotent semirings $\Semi$.
\end{restatable}
	
\begin{proof}
	Suppose that $\Semi$ is $n$-idempotent and let $\varphi(\bar{x}) = Qx \psi (\bar{x},x)$ separate $(\pi_A, \bar{a})$ and $(\pi_B, \bar{b})$, i.e.
	$\sum_{a \in A} \pi_A \llb \psi (\bar{a},a) \rrb \neq \sum_{b \in B} \pi_B \llb \psi (\bar{b},b) \rrb$ or
	$\prod_{a \in A} \pi_A \llb \psi (\bar{a},a) \rrb \neq \prod_{b \in B} \pi_B \llb \psi (\bar{b},b) \rrb$.
	Due to associativity and $n$-idempotence, there must be some $s \in S$ such that
	\begin{align*}
		|\underbrace{\{a \in A \colon \pi_A \llb \psi(\bar{a}, a) \rrb = s\}}_{=: A_{\psi, \bar{a}}^s}| \neq |\underbrace{\{b \in B \colon \pi_B \llb \psi(\bar{b}, b) \rrb = s\}}_{=: B_{\psi, \bar{b}}^s}|,
	\end{align*}
	where $|A_{\psi, \bar{a}}^s| < n$ or $|B_{\psi, \bar{b}}^s| < n$. Assume w.l.o.g. that $|A_{\psi, \bar{a}}^s| < |B_{\psi, \bar{b}}^s|$.
	Spoiler wins the game $CG_{\qr(\varphi)}^n(\pi_A, \bar{a}, \pi_B, \bar{b})$ as follows.
	First, he chooses some $B' \subseteq B_{\psi, \bar{b}}^s$ with $|B'| = |A_{\psi, \bar{a}}^s|+1 \leq n$.
	For Duplicator's answer $A' \subseteq A$, there must be some $a \in A'$ such that $\pi_A \llb \psi (\bar{a}, a) \rrb \neq s$, since $|A'| = |B'| > |A_{\psi, \bar{a}}^s|$.
	Spoiler picks this element~$a$. Regardless of Duplicator's response $b \in B_{\psi, \bar{b}}^s$, it holds that $\pi_A \llb \psi (\bar{a}, a) \rrb \neq s = \pi_B \llb \psi (\bar{b}, b) \rrb$ and Spoiler wins the remaining subgame by induction.
	
	If $\Semi$ is not $n$-idempotent, there must be index sets $I,J$ with $|I|, |J| \geq n$ such that $\sum_{i \in I} s \neq \sum_{j \in J} s$ or $\prod_{i \in I} s \neq \prod_{j \in J} s$.
	Let $\pi_A$ and $\pi_B$ be $\Semi$-interpretations over~$\{R\}$ with universes $A \coloneqq \{a_i \colon i \in I\}$ and $B \coloneqq \{b_j \colon j \in J \}$ satisfying $\pi_A \llb R a \rrb = \pi_B \llb R b \rrb = s$ and $\pi_A \llb \lnot R a \rrb = \pi_B \llb \lnot R b \rrb = 0$ for $a \in A, b \in B$, where $R$ is a unary predicate. Every strategy for Duplicator in $CG_1^n (\pi_A, \pi_B)$ is winning, but $\pi_A$ and $\pi_B$ are separable by $\exists x Rx$ or $\forall x Rx$.
\end{proof}
	
\subsection{Completeness and incompleteness}
	
As opposed to a Boolean quantifier or a move in an Ehrenfeucht--Fraïssé game,
a quantifier in semiring semantics does not pick out a specific element of the universe.
Instead, it induces a sum or product over all elements.
As a consequence, completeness of the $m$-turn Ehrenfeucht--Fraïssé game, and thus also completeness of other variants of model comparison games, fail in general.
In particular, this applies to semirings on which elementary equivalence of finite interpretations does not imply isomorphism.
Indeed, on any pair of finite non-isomorphic semiring interpretations, Spoiler wins $G_m$ for sufficiently large $m$ by picking all elements in the larger universe.
A particular example, presented in \cite{GraedelMrk21}, of non-isomorphic but elementarily equivalent $\Semi$-interpretations $\pi_A^{s,t}$ and $\pi_B^{s,t}$
for arbitrary elements $s, t$ of a fully idempotent semiring $\Semi$ is the following:

\begin{center}
	$\pi_A^{s,t}:$ 
	\begin{tabular}{c||c|c|c|c}
		$A$ & $R_1$ & $R_2$ & $\lnot R_1$ & $\lnot R_2$ \\
		\hline
		\hline
		$a_1$ & $0$ & $t$ & $s$ & $0$ \\
		\hline
		$a_2$ & $s$ & $0$ & $0$ & $t$ \\
		\hline
		$a_3$ & $t$ & $s$ & $0$ & $0$ \\
		\hline
		$a_4$ & $0$ & $0$ & $t$ & $s$
	\end{tabular}
	\hspace{1cm}
	$\pi_B^{s, t}:$ 
	\begin{tabular}{c||c|c|c|c}
		$B$ & $R_1$ & $R_2$ & $\lnot R_1$ & $\lnot R_2$ \\
		\hline
		\hline
		$b_1$ & $t$ & $0$ & $0$ & $s$ \\
		\hline
		$b_2$ & $0$ & $s$ & $t$ & $0$ \\
		\hline
		$b_3$ & $s$ & $t$ & $0$ & $0$ \\
		\hline
		$b_4$ & $0$ & $0$ & $s$ & $t$ \\
	\end{tabular}
\end{center}
	
For any $s,t\in \Semi$, we have that  $\pi_A^{s,t}\equiv \pi_B^{s,t}$ \cite[Theorem 13]{GraedelMrk21}, but obviously,  Spoiler even wins the game 
$G_1 (\pi_A^{s,t}, \pi_B^{s, t})$ for distinct and non-zero values $s,t \in S$.
Thus, completeness of~$G_m$ for~$\equiv_m$ and full idempotence
are mutually exclusive on semirings with at least three elements,
while soundness requires full idempotence,
which entails the following result.
	
\begin{theorem} \label{thm-complIdem}
If, for all $m\in\N$, the game $G_m$ is sound and complete for $\equiv_m$ on $\Semi$,
then~$\Semi$ is isomorphic to $\mathbb{B}$.
\end{theorem}

Several further semirings, such as $\mathbb{L}, \WX, \mathbb{S}[X]$ or $\mathbb{B}[X]$,
admit pairs of finite interpretations that are non-isomorphic but elementarily equivalent,
which immediately disproves completeness of $G_m$ for $\equiv_m$
on those semirings (see \cref{sum-Gm-BGm}).
Moreover, even on semirings such as $\Trop$, $\N$ and $\N[X]$,
where it is known that elementary equivalence does coincide with isomorphism
on finite interpretations \cite{GraedelMrk21},
$G_m$ is not necessarily complete.
As a counterexample on the tropical semiring $\Trop = (\R^\infty_+ , \min, +, \infty, 0)$,
consider the following $\Trop$-interpretations.

\medskip

\begin{minipage}[c]{0.9\linewidth}
\centering
$\pi_A:\quad$
\begin{tabular}{c || c | c}
$A$ & $R$ & $\neg R$ \\ \hline\hline
$a_0$ & 0 & $\infty$ \\
$a_1$ & 1 & $\infty$ \\
$a_2$ & 1 & $\infty$ \\
\end{tabular}
$\quad\quad\quad\pi_B:\quad$
\begin{tabular}{c || c | c}
$B$ & $R$ & $\neg R$ \\ \hline\hline
$b_0$ & 0 & $\infty$ \\
$b_1$ & 2 & $\infty$ \\
\end{tabular}
\end{minipage}

\medskip

Clearly, Spoiler already wins $G_1(\pi_A, \pi_B)$,
but we can show that~$\pi_A \equiv_1 \pi_B$,
thus, the game~$G_1$ is incomplete for~$\equiv_1$ on~$\Trop$.
The $1$-equivalence immediately follows
from the following criterion.
	
\begin{proposition} \label{prop-V1equiv}
	Two $\Trop$-interpretations $\pi_A, \pi_B$ over vocabulary $\tau= \{R\}$, consisting of a single unary relation symbol, are $1$-equivalent if
	\begin{bracketenumerate}
		\item $\pi_A (\neg R a) = \pi_B (\neg R b) = \infty$ for all $a \in A$ and $b \in B$,
		\item $\min_{a \in A} \pi_A (R a) = \min_{b \in B} \pi_B (R b)$ and
		\item $\sum_{a \in A} \pi_A (R a) = \sum_{b \in B} \pi_B (R b)$.
	\end{bracketenumerate}
\end{proposition}

\begin{proof}
	Let $\Phi := \{x=x, x \neq x\} \cup \{(Rx)^n \colon n \in \mathbb{N}_{>0}\}$ where $(Rx)^n$ denotes the $n$-fold conjunction of  $Rx$.
	It follows by induction that for all quantifier-free formulae $\psi(x)$, there is some $\psi^*(x) \in \Phi$ such that $\pi_A \llb \psi (a) \rrb = \pi_A \llb \psi^* (a) \rrb$ for all $a \in A$ and $\pi_B \llb \psi (b) \rrb = \pi_B \llb \psi^* (b) \rrb$ for all $b \in B$.
	If there is a sentence of quantifier rank $1$ which separates $\pi_A$ and $\pi_B$, then there must be a separating sentence of the form $Q x \psi (x)$ with $Q \in \{\exists, \forall \}$.
	But this implies that~$\pi_A$ and~$\pi_B$ can be separated by some formula $Q x \psi^* (x)$ where $\psi^* (x) \in \Phi$. This yields a contradiction, since 
	$\min_{a \in A} \pi_A \llb \psi^* (a) \rrb = \min_{b \in B} \pi_B \llb \psi^* (b) \rrb$ and $\sum_{a \in A} \pi_A \llb \psi^* (a) \rrb = \sum_{b \in B} \pi_B \llb \psi^* (b) \rrb$,
	so a separating formula of quantifier rank $1$ cannot exist, i.e., $\pi_A \equiv_1 \pi_B$.
\end{proof}

On the other side, in contrast to the classical $m$-turn Ehrenfeucht--Fraïssé game, there are semirings other than $\Bool$, namely $\N$ and $\N[X]$,
on which the $m$-turn bijection game is both sound and complete.
While soundness of $BG_m$ holds on any semiring, we prove completeness on $\N$ by refining the characteristic 
sentences from $\cite{GraedelMrk21}$ and making use of the following combinatorial lemma.

\begin{lemma} \label{lem-complN}
	For each $\ell, d \in \N$, there is a sufficiently large $e \in \N$ such that for all $(r_1, \dots r_{\ell'}), (s_1, \dots , s_{\ell'}) \in \N^{\ell'}$ with $\ell' < \ell$ and $r_i, s_i < d$ for $1 \leq i \leq \ell'$,
	\[
	\sum\limits_{i = 1}^{\ell'} r_i^e = \sum\limits_{i=1}^{\ell'} s_i^e
	\]
	implies that there is a permutation $\sigma \in S_{\ell'}$ such that $r_i = s_{\sigma(i)}$ for all $1 \leq i \leq  \ell'$.
\end{lemma}

In the classical  Ehrenfeucht--Fraïssé theory, characteristic sentences~$\chi^m_\AA$ for a structure~$\AA$ (with quantifier rank $m$) 
formulate the existence of $m$-back-and-forth systems in the sense that $\BB\models\chi^m_\AA$ if, and only if,
such a system exists for $\AA$ and $\BB$. In the context of semiring interpretations, there are somewhat similar, but slightly more involved constructions
to describe certain semiring interpretations up to $m$-equivalence.
Our goal here is to construct characteristic formulae $\chi_{\bar{c}}^m(x_1, \dots, x_n)$ of quantifier rank $m$ depending on a pair of constants $\bar{c} = (c_1, c_2) \in \N^2$ such that $\pi_A \llb \chi_{\bar{c}}^m(\bar{a}) \rrb = \pi_B \llb \chi_{\bar{c}}^m(\bar{b}) \rrb$ ensures that Duplicator wins $BG_m (\pi_A , \bar{a}, \pi_B ,  \bar{b})$ if $\pi_A$ and $\pi_B$ only include valuations smaller than $c_1$ and if their universes are of cardinality less than $c_2$.
For this purpose, we use auxiliary formulae $\vartheta_{\bar{c}}^m(x_1, \dots, x_n)$, which are meant to ensure Duplicator's victory in $BG_m$ under the assumption that $|A| = |B|$.
Accordingly, $\vartheta_{\bar{c}}^0(\bar{x})$ shall characterise the winning condition of the bijection game.
To implement this, we fix an enumeration $L_1 (\bar{x}), \dots, L_k (\bar{x})$ of the $\tau$-literals in $\lit_n(\tau)$ and represent the valuations of the $\tau$-literals as digits in a number system. Choosing the radix large enough ensures that the single valuations coincide in $\pi_A$ and $\pi_B$. 
$$
\vartheta_{\bar{c}}^0 (\bar{x})  := \bigvee\limits_{1 \leq i \leq k} \underbrace{(L_i (\bar{x}) \vee \dots \vee L_i (\bar{x}))}_{c_1^{i-1} \text{ times}}
$$
Based on $\vartheta^{m-1}_{\bar{c}} (\bar{x},x)$, we define $\vartheta^{m}_{\bar{c}} (\bar{x})$ such that $\pi_A \llb \vartheta^m_{\bar{c}}(\bar{a}) \rrb = \pi_B \llb \vartheta^m_{ \bar{c}}(\bar{b}) \rrb$ ensures that $(\pi_A \llb \vartheta_{ \bar{c}}^{m-1}(\bar{a},a) \rrb)_{a \in A \setminus \{a_1, \dots, a_n\}}$ and $(\pi_B \llb \vartheta_{\bar{c}}^{m-1}(\bar{b},b) \rrb)_{b \in B \setminus \{b_1, \dots, b_n\}}$ only differ by some permutation according to \cref{lem-complN}. Let
$$
\vartheta_{\bar{c}}^m (\bar{x}) := \exists x ((\bigwedge\limits_{1 \leq i \leq n} x \neq x_i \wedge \vartheta_{ \bar{c}}^{m-1}(\bar{x},x))^{e_{m-1}}),
$$
where $e_{m-1}$ is chosen according to \cref{lem-complN} with respect to $\ell:=\max(c_2,4)$ and $d_{m-1}$ which is inductively defined by $d_0 := c_1^{k+1}$ and $d_{i+1} := c_2 \cdot d_{i}^{e_{i}}$ for $i>0$. Note that this definition ensures that $d_m > \pi_A \llb \vartheta^m_{\bar{c}} (\bar{a}) \rrb$ and $d_m > \pi_B \llb \vartheta^m_{\bar{c}} (\bar{b}) \rrb$ for all $m \in \N$.

In order to drop the assumption $|A|=|B|$ the formulae $\vartheta^m_{\bar{c}}(\bar{x})$ rely on, since \cref{lem-complN} presumes tuples of the same length, we additionally encode in $\chi^m_{\bar{c}} (\bar{x})$ that the universes must be of the same cardinality. Having defined the sequence $(e_{m})_{m \in \N}$ of exponents with respect to tuples of length smaller than $\max(c_2, 4)$ allows us to reuse them for this purpose. Let $\chi_{\bar{c}}^0 (\bar{x}) := \vartheta_{\bar{c}}^0 (\bar{x})$ and, for $m>0$,
$$
\chi_{\bar{c}}^m (\bar{x}) := (\exists x (x=x) \vee \exists x (x=x) \vee \vartheta^m_{\bar{c}} (\bar{x}) )^{e_{m}}.
$$

\begin{theorem} \label{thm-complN}
	Let $c_1,c_2,m \in \N$. For all (finite) $\N$-interpretations $\pi_A$, $\pi_B$ and $\bar{a} \in A^n$, $\bar{b} \in B^n$ such that $\max (\img (\pi_A) \cup \img (\pi_B)) < c_1$ and $|A| < c_2$, $|B| < c_2$, the following are equivalent:
	\begin{bracketenumerate}
		\item Duplicator wins $BG_m(\pi_A, \bar{a}, \pi_B, \bar{b})$;
		\item $\pi_A \llb \chi^m_{\bar{c}} (\bar{a}) \rrb = \pi_B \llb \chi^m_{\bar{c}} (\bar{b}) \rrb$;
		\item $(\pi_A, \bar{a}) \equiv_m (\pi_B, \bar{b})$.
	\end{bracketenumerate}
\end{theorem}

\begin{proof}
	By \cref{thm-soundBGm}, it holds that $(1) \Rightarrow (3)$. Since $\qr (\chi^m_{\bar{c}}(\bar{x}))=m$ by definition, we can immediately infer implication $(3) \Rightarrow (2)$.
	It remains to show $(2) \Rightarrow (1)$, which we prove by induction on $m$ for all tuples $\bar{a} \in A^n$ and $\bar{b} \in B^n$ simultaneously.
	To this end, let $\pi_A, \pi_B$ and~$\bar{c}=(c_1,c_2)$ be given as above.
	
	In the base case $m=0$, as $\pi_A (L_i (\bar{a})) < c_1$ and $\pi_B (L_i (\bar{b})) < c_1$ for all $1 \leq i \leq k$ by assumption, $\pi_A \llb \chi^0_{\bar{c}} (\bar{a}) \rrb = \pi_B \llb \chi^0_{\bar{c}} (\bar{b}) \rrb$ implies that $\pi_A (L_i (\bar{a}))= \pi_B (L_i(\bar{b}))$ for all $1 \leq i \leq k$, which is why Duplicator wins the game $BG_0(\pi_A, \bar{a}, \pi_B, \bar{b})$.
	
	Now, for the induction step, we assume $m>0$. Due to the definition of $d_m$, it holds that $d_m \geq c_2 > \max(|A|,|B|) = \max(\pi_A \llb \exists x (x=x) \rrb, \pi_B \llb \exists x (x=x) \rrb)$.
	Further, ${d_{m} > \max(\pi_A \llb \vartheta^m_{\bar{c}} (\bar{a}) \rrb, \pi_B \llb \vartheta^m_{\bar{c}} (\bar{b}) \rrb)}$. Since $e_{m}$ was chosen with respect to $\ell= \max(c_2,4)> 3$ and $d_{m}$, we infer by \cref{lem-complN} that the triples
	$(\pi_A \llb \exists x (x=x) \rrb, \pi_A \llb \exists x (x=x) \rrb, \pi_A \llb \vartheta^{m}_{\bar{c}} (\bar{a}) \rrb)$ and
	$(\pi_B \llb \exists x (x=x) \rrb, \pi_B \llb \exists x (x=x) \rrb, \pi_B \llb \vartheta^{m}_{\bar{c}} (\bar{b}) \rrb)
	$
	only differ by some permutation $\sigma$. Suppose that $\sigma$ is not the identity mapping.
	Then, $\pi_A \llb \exists x (x=x) \rrb = \pi_B \llb \exists x (x=x) \rrb = \pi_B \llb \vartheta^{m}_{\bar{c}} (\bar{b}) \rrb$ and $\pi_A \llb \vartheta^{m}_{\bar{c}} (\bar{a}) \rrb = \pi_B \llb \exists x (x=x) \rrb$ follows, which also implies that the tuples have to coincide.
	Hence, we can conclude that $\pi_A \llb \exists x (x=x) \rrb = \pi_B \llb \exists x (x=x) \rrb$, which is equivalent to $|A|=|B|$, and $\pi_A \llb \vartheta^m_{\bar{c}} (\bar{a}) \rrb =\pi_B \llb \vartheta^m_{\bar{c}} (\bar{b}) \rrb$.
	By definition, the latter implies that
	$$
	\sum\limits_{a \in A\setminus \{a_1, \dots, a_n\}} \pi_A \llb \vartheta^{m-1}_{\bar{c}}(\bar{a},a) \rrb^{e_{m-1}} = \sum\limits_{b \in B\setminus \{b_1, \dots, b_n\}} \pi_B \llb \vartheta^{m-1}_{\bar{c}}(\bar{b},b) \rrb^{e_{m-1}}.
	$$
	Since $\ell > |A|-n=|B|-n$ and $d_{m-1} > \max (\pi_A \llb \vartheta^{m-1}_{\bar{c}}(\bar{a},a) \rrb,\pi_B \llb \vartheta^{m-1}_{\bar{c}}(\bar{b},b) \rrb)$ for all elements $a \in A \setminus \{a_1, \dots, a_n\}$ and $b \in B \setminus \{b_1, \dots, b_n \}$, by \cref{lem-complN} there is a bijection $h \colon A \setminus \{a_1, \dots, a_n\} \to B \setminus \{b_1, \dots, b_n\}$ such that 
	$
	\pi_A \llb \vartheta^{m-1}_{\bar{c}}(\bar{a},a) \rrb =  \pi_B \llb \vartheta^{m-1}_{\bar{c}}(\bar{b}, h(a)) \rrb
	$
	for all $a \in A \setminus \{a_1, \dots, a_n\}$, which implies $\pi_A \llb \chi^{m-1}_{\bar{c}}(\bar{a},a) \rrb =  \pi_B \llb \chi^{m-1}_{\bar{c}}(\bar{b}, h(a)) \rrb$ with $|A|=|B|$. 
	Duplicator can win the game $BG_m (\pi_A, \bar{a}, \pi_B, \bar{b})$ as follows: She provides the bijection $h'$ which extends $h$ to the domain $A$ by $h'(a_i) = b_i $ for all $1\leq i \leq n$.
	W.l.o.g. we can assume that Spoiler picks some $a \in A \setminus \{a_1, \dots, a_n\}$.
	We obtain for the updated position $(\bar{a},a, \bar{b},h(a))$ that $\pi_A \llb \chi^{m-1}_{\bar{c}}(\bar{a},a) \rrb =  \pi_B \llb \chi^{m-1}_{ \bar{c}}(\bar{b}, h(a)) \rrb$ must hold.
	By induction hypothesis, Duplicator has a strategy to win the remaining subgame $BG_{m-1}(\pi_A, \bar{a}, \pi_B, h(a))$. 
\end{proof}

By invoking the following lemma, we can transfer our result concerning $\N$-interpretations to the semiring of polynomials $\N [X]$.

\begin{lemma} \label{lem-complNX}
	Let $\N[X](c,e) \subseteq \N[X]$ denote the set of polynomials with coefficients less than $c$ and exponents smaller than $e$. There is a variable assignment $\alpha \colon X \to \N$ inducing a homomorphism $h \colon \N[X] \to \N$ whose restriction $h\vert_{\N[X](c,e)}$ is a bijection from $\N[X](c,e)$ to~$\{0,\dots, c^{e^{|X|}}-1\}$.
\end{lemma}

\begin{theorem} \label{thm-complNX}
	Let $c,e,c_2,m \in \N$. For all finite $\N[X]$-interpretations $\pi_A$, $\pi_B$ and $\bar{a} \in A^n$, $\bar{b} \in B^n$ such that $\img (\pi_A) \cup \img (\pi_B) \subseteq \N[X](c,e)$ and $|A| < c_2$, $|B| < c_2$, the following are equivalent:
	\begin{bracketenumerate}
		\item Duplicator wins $BG_m(\pi_A, \bar{a}, \pi_B, \bar{b})$;
		\item $\pi_A \llb \chi^m_{(c_1, c_2)} (\bar{a}) \rrb = \pi_B \llb \chi^m_{(c_1, c_2)} (\bar{b}) \rrb$ where $c_1:= c^{e^{|X|}}$;
		\item $(\pi_A, \bar{a}) \equiv_m (\pi_B, \bar{b})$.
	\end{bracketenumerate}
\end{theorem}

\begin{proof}
	Following the same reasoning as in \cref{thm-complN}, it suffices to prove $(2) \Rightarrow (1)$. Let $h \colon \N[X] \to \N$ be a homomorphism according to \cref{lem-complNX}.
	Due to the fundamental property, $(2)$ implies ${(h \circ \pi_A) \llb \chi^m_{\bar{c}} (\bar{a}) \rrb =  (h \circ \pi_B) \llb \chi^m_{\bar{c}} (\bar{b}) \rrb}$.
	Further, it must hold that $\max (\img(h \circ \pi_A) \cup \img (h \circ \pi_A)) < c_1$, because of \cref{lem-complNX} and the assumption that $\img (\pi_A) \cup \img (\pi_B) \subseteq \N[X](c,e)$.
	By \cref{thm-complN}, this implies that Duplicator has a winning strategy for $BG_m(h \circ \pi_A, \bar{a}, h \circ \pi_B, \bar{b})$.
	The strategy ensures that any reachable position $(\bar{a}, a_{n+1}, \dots, a_{n+m}, \bar{b}, b_{n+1}, \dots, b_{n+m})$ induces a local isomorphism between $\pi_A$ and $\pi_B$, i.e., for each $L (\bar{x}) \in \lit_{n+m} (\tau)$ it holds that $h \circ \pi_A (L(a_1, \dots, a_{n+m})) = h \circ \pi_B (L(b_1, \dots, b_{n+m})))$.
	Due to injectivity of $h\vert_{\N[X](c,e)}$, we can derive $\pi_A (L(a_1, \dots, a_{n+m})) = \pi_B (L(b_1, \dots, b_{n+m}))$ for all $L(\bar{x}) \in \lit_{n+m} (\tau)$. Hence, the strategy must also be winning for Duplicator in the game $BG_m(\pi_A, \bar{a}, \pi_B, \bar{b})$.
\end{proof}

\begin{corollary} \label{thm-soundcomplNNX}
	For every $m\in\N$, the bijection game $BG_m$ is sound and complete for $\equiv_m$ on $\N$ and $\N[X]$.
\end{corollary}

\Cref{fig:landscape} summarises the results from this section
by subdividing pairs of $\Semi$-interpretations
based on $m$-equivalence and the outcomes of
model comparison games~$G_m$, $CG_m^n$~and~$BG_m$.
Soundness of the Ehrenfeucht--Fraïssé game~$G_m$ on~$\Semi$ holds
whenever the lower right gray quadrant is empty,
and completeness holds whenever the upper left gray quadrant is empty.
Results for counting games and bijection games,
where the winning condition for Duplicator is more restrictive in general,
are illustrated in the lower left corner.

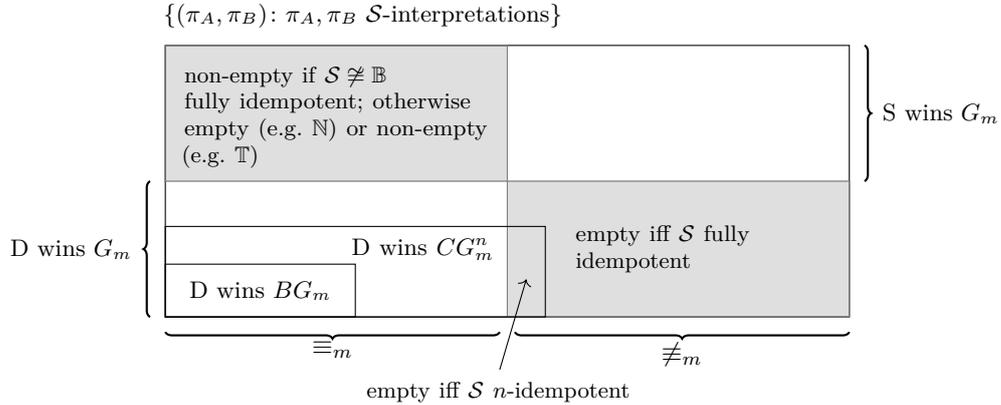
\begin{figure}[h]
	\centering
	\begin{tikzpicture}
	\tikzstyle{bracket}=[decorate, decoration={brace}, thick]
	
	\draw (0,0) rectangle ++(9,3.6);
	\draw[fill=gray, opacity=.25] (4.5,0) rectangle ++(4.5,1.8);
	\draw[fill=gray, opacity=.25] (0,1.8) rectangle ++(4.5,1.8);
	\draw[draw=gray] (4.5,0) rectangle ++(4.5,1.8);
	\draw[draw=gray] (0,1.8) rectangle ++(4.5,1.8);
	
	\draw (0,0) rectangle node {\small D wins $BG_m$} ++(2.5,.7);
	\draw (0,0) rectangle ++(5,1.2);
	\node[text width=2.5cm, execute at begin node=\setlength{\baselineskip}{1em}] at (3.7,.9) {\small{D wins $CG_m^n$}};
	%\node[text width=1cm, execute at begin node=\setlength{\baselineskip}{1em}] at (5,.7) {\footnotesize{empty iff $\Semi$ $\kappa$-idempotent}};
	
	\draw[bracket] (-.2,0) -- node[left] {\small D wins $G_m$ \ } (-.2,1.8);
	\draw[bracket] (9.2,3.6) -- node[right] {\small \ S wins $G_m$} (9.2,1.8);
	\draw[bracket] (4.4,-.2) -- node[below] {$\equiv_m$} (0,-.2);
	\draw[bracket] (9,-.2) -- node[below] {$\not\equiv_m$} (4.6,-.2);
	
	\node[text width=2.6cm, execute at begin node=\setlength{\baselineskip}{1em}] at (6.7,.9) {\footnotesize{empty iff $\Semi$ fully idempotent}};
	\node[text width=4.5cm, execute at begin node=\setlength{\baselineskip}{1em} ] at (2.5,2.65) {\footnotesize{non-empty if $\Semi \not\cong \Bool$ \\ fully idempotent; otherwise empty (e.g.\ $\N$) or non-empty (e.g.\ $\Trop$)}};
	
	\node at (2.6,4) {\small $\{(\pi_A, \pi_B) \colon \pi_A, \pi_B \ \Semi\text{-interpretations}\}$};
	
	\draw[->] (4.4,-0.7) to (4.75,0.5);
	\node[text width=3.5cm, execute at begin node=\setlength{\baselineskip}{1em}] at (4.4,-1) {\footnotesize{empty iff $\Semi$ $n$-idempotent}};
\end{tikzpicture}
	\caption{Landscape of pairs of $\Semi$-interpretations,
	subdivided by $m$-equivalence and outcomes of model comparison games.}
	\label{fig:landscape}
\end{figure}
	
	\section{Characterising elementary equivalence} \label{sec:ElEquiv}
	
The Ehrenfeucht--Fraïssé theorem also provides a game-theoretic characterisation of elementary equivalence via the game $G(\AA, \BB)$ where Spoiler chooses the number of turns at the beginning of each play. We now discuss soundness and completeness of $G$ for $\equiv$ on semirings. For classical structures, soundness and completeness of $G$ for $\equiv$ is equivalent to soundness and completeness of $G_m$ for $\equiv_m$, for all $m$, 
but this is in general not the case on semirings.
	
For the study of the game $G$, interpretations on infinite universes are of particular interest.
This especially applies to soundness, which is trivial in the finite case,
since a winning strategy for Duplicator already implies isomorphism on finite interpretations.
Semiring semantics for infinite interpretations requires sum and product operators on infinite families $(s_i)_{i \in I} \subseteq S$ of semiring elements.
There are certain semirings such as $\N, \N[X], \Bool [X]$ and $\Sorb[X]$ which do not admit a reasonable definition of such infinitary operations,
and we thus have to restrict ourselves to finite universes.
Otherwise, we make use of the natural order and interpret infinite sums according to $\sum_{i \in I} s_i := \sup \{\sum_{i \in I'} s_i | I' \subseteq I \text{ finite}\}$.
For infinitary products we distinguish the case of absorptive semirings, where multiplication is decreasing and we thus
interpret the product as the \emph{infimum} of the finite subproducts, and the cases, such as $\Ninf$ or $\WX$, where multiplication is increasing and we
replace infima by suprema.
Previous results such as the soundness of $G_m$ on fully idempotent semirings straightforwardly extend to infinite interpretations by transferring semiring properties such as full idempotence to the infinitary operations.
	
\subsection{Soundness of the game \texorpdfstring{$G$}{G}}  \label{subsec:soundG}

Soundness of~$G$ for~$\equiv$ holds whenever Spoiler wins $G(\pi_A, \pi_B)$
for all first-order separable interpretations~$\pi_A$ and~$\pi_B$.
Thus, the following question is essential:
Given $\pi_A$, $\pi_B$ and a separating sentence~$\psi$,
is the required number of turns for Spoiler to win $G(\pi_A, \pi_B)$ bounded in advance?
On fully idempotent semirings,
$m \coloneqq \qr(\psi)$ turns suffice for Spoiler to win $G(\pi_A, \pi_B)$ since~$G_m$ is sound for~$\equiv_m$,
which immediately implies soundness of~$G$ on all fully idempotent semirings.
However, full idempotence is not a necessary condition,
soundness of~$G$ is still preserved
on many semirings that admit a weaker bound than~$m$.
For instance, on any $n$-idempotent semiring for some $n \in \N$,
$n \cdot m$ turns suffice to ensure Spoiler's victory.
	
\begin{proposition} \label{thm-soundWXG}
	Let $\Semi$ be $n$-idempotent for some $n \in \N$. For any $\Semi$-interpretations $\pi_A$ and $\pi_B$ it holds that $\pi_A \equiv_m \pi_B$ if Duplicator wins the game $G_{nm} (\pi_A, \pi_B)$. In particular, the game $G$ is sound for $\equiv$ on $\Semi$.
\end{proposition}

This follows from soundness
of $n$-counting games as stated in \cref{thm-soundCGm}.
If $\pi_A \not\equiv_m \pi_B$, Spoiler wins $G_{nm}(\pi_A, \pi_B)$
by adapting his winning strategy for $CG_m^n(\pi_A, \pi_B)$:
Instead of drawing $n$-element sets, he draws~$n$ elements one by one.
Note that the bound $n \cdot m$ does not depend on~$\pi_A$ and~$\pi_B$ at all,
but only on the quantifier rank $m$ and the semiring.

However, other semirings, such as $\N^{\infty}$,
may not admit an inherent bound $t(m) \in \N$
such that  a winning strategy of Duplicator for
$G_{t(m)}(\pi_A, \pi_B)$ always implies $\pi_A \equiv_m \pi_B$.
To demonstrate this, consider a pair of sets
($\N^{\infty}$-interpretations with empty vocabulary)
with~$t(m)$ and $t(m) + 1$ elements, respectively.
Clearly, Duplicator wins on those sets for up to $t(m)$ turns,
but the sentence $\exists x (x = x)$ with quantifier rank~$1$ suffices to separate them.

In order to prove that the game~$G$ is still sound for~$\equiv$ on $\N^{\infty}$,
it is crucial to observe that two separable interpretations~$\pi_A, \pi_B$
with $\pi_A \llb \psi \rrb \neq \pi_B \llb \psi \rrb$
admit a parameter~$k$ that induces an upper bound
on the number of moves required by Spoiler to win $G(\pi_A, \pi_B)$.
On~$\mathbb{N}^{\infty}$, this parameter is easily obtained by observing
that $\pi_A \llb \psi \rrb$ or $\pi_B \llb \psi \rrb$ is finite.
	
Two first-order separable $\N^{\infty}$-interpretations
$\pi_A, \pi_B$ admit a separating sentence~$\psi$.
Hence, soundness of~$G$ on~$\N^{\infty}$
follows directly from \cref{thm-soundNinfG}
with $k \coloneqq \min \{\pi_A \llb \psi \rrb, \pi_B \llb \psi \rrb\} + 1$.

\begin{restatable}{theorem}{THMsoundNinfG} \label{thm-soundNinfG}
	Let $\pi_A$ and $\pi_B$ be $\mathbb{N}^\infty$-interpretations with elements
	$\bar{a} \in A^n$, $\bar{b} \in B^n$ and $k \geq 1$. If there is a separating formula $\varphi(\bar{x})$ with $\qr (\varphi(\bar{x})) \leq m$ such that $\pi_A \llb \varphi(\bar{a}) \rrb < k$ or $\pi_B \llb \varphi(\bar{b}) \rrb < k$, then Spoiler wins $G_{k m}(\pi_A, \bar{a}, \pi_B, \bar{b})$.
\end{restatable}

\begin{proof}
	We proceed by induction on $\varphi(\bar{x})$.
	If $\varphi(\bar{x})$ is a literal the claim holds trivially.
	
	Let $\varphi(\bar{x}) = \varphi_1(\bar{x}) \vee \varphi_2 (\bar{x})$. W.l.o.g. let $\pi_A \llb \varphi(\bar{a}) \rrb < k$.
	Since addition is increasing in $\mathbb{N}^\infty$, we have that $\pi_A \llb \varphi_1(\bar{a}) \rrb < k$ and $\pi_A \llb \varphi_2(\bar{a}) \rrb < k$. Moreover, $\pi_A \llb \varphi (\bar{a}) \rrb \neq \pi_B \llb \varphi (\bar{b}) \rrb$ implies $\pi_A \llb \varphi_i (\bar{a}) \rrb \neq \pi_B \llb \varphi_i (\bar{b}) \rrb$ for some $i \in \{1,2\}$. By induction hypothesis, it follows that Spoiler wins $G_{k  m}(\pi_A, \bar{a}, \pi_B, \bar{b})$. 
	The case $\varphi(\bar{x}) = \varphi_1(\bar{x}) \wedge \varphi_2 (\bar{x})$ is analogous.
	
	Now, consider $\varphi(\bar{x}) = \exists y \psi(\bar{x}, y)$ and
	suppose w.l.o.g. that $\pi_A \llb \varphi(\bar{a}) \rrb < \pi_B \llb \varphi(\bar{b}) \rrb$. 
	Let $A' := \{ a \in A \colon \pi_A \llb \psi(\bar{a},a) \rrb > 0 \}$. 
	Clearly, $|A'| < k$, since $\pi_A \llb \varphi(\bar{a}) \rrb < k$ by assumption. 
	In the game $G_{k m}(\pi_A, \bar{a}, \pi_B, \bar{b})$, Spoiler successively draws all elements $a \in A'$.
	If Duplicator manages to find for each $a \in A'$ a unique duplicate $b \in B$ such that $\pi_A \llb \psi (\bar{a}, a) \rrb = \pi_B \llb \psi (\bar{b}, b) \rrb$, then there must be an $(|A'|+1)$-th element in $b \in B$ with $\pi_B \llb \psi (\bar{b}, b) \rrb > 0$, because $\pi_A \llb \varphi(\bar{a}) \rrb < \pi_B \llb \varphi(\bar{b}) \rrb$.
	Hence, if Duplicator was able to duplicate all previous choices, Spoiler additionally chooses such an element $b \in B$ afterwards.
	In any case, after at most~$k$ turns, a pair~$(a, b)$ was picked such that
	$\pi_A \llb \psi(\bar{a}, a) \rrb \neq \pi_B \llb \psi(\bar{b}, b) \rrb$.
	Since $\pi_A \llb \varphi(\bar{a}) \rrb < k$ by assumption, it holds that $\pi_A \llb \psi(\bar{a},a) \rrb < k$ for all $a \in A$ and the induction hypothesis can be applied to $\psi(\bar{x}, x)$ with instantiations $(\bar{a}, a)$ and $(\bar{b}, b)$. We obtain that Spoiler wins the game $G_{k(m-1)}(\pi_A, \bar{a}, a, \pi_B, \bar{b}, b)$,
	hence, he wins the remaining subgame.
	The case for universally quantified formulae $\varphi(\bar{x}) = \forall x \psi(\bar{x}, y)$ is treated similarly with slight modifications, since multiplication is also increasing on~$\N^{\infty}$ (if we exclude~$0$).
\end{proof}

It turns out that a similar approach
is applicable to the semiring $\Sinf[X]$,
which extends the semiring $\mathbb{S}[X]$
of absorptive polynomials to allow infinite exponents (and thus infinite products),
albeit the derivation of a suitable parameter is more involved.
Recall that a monomial~$m$ \emph{absorbs} a monomial~$m'$
if the exponents satisfy $m(x) \le m'(x)$ for all $x \in X$
and that absorptive polynomials only retain absorption-dominant monomials.
We say that a monomial~$m$ \emph{separates} polynomials~$p$ and~$q$
if $m \in p$ and~$m$ is not absorbed by any monomial from~$q$.

These concepts can be extenuated to any subset $Y \subseteq X$: 
$m$ $Y$-absorbs~$m'$ iff $m(x) \le m'(x)$ for $x \in Y$,
and it is $Y$-separating for~$p$ and~$q$
if it is contained in one of the polynomials
but not $Y$-absorbed by any of the monomials from the other polynomial.
Finally, we can parametrise monomials~$m$
by adding their exponents $e_Y(m) \coloneqq \sum_{x \in Y} m(x)$
for all the variables $x \in Y$.
Now, we can extract a finite parameter
from any pair of distinct polynomials~$p, q$ as follows.

\begin{restatable}{lemma}{LEMsoundSinfG} \label{lem-soundSinfG}
For any two distinct polynomials $p, q \in \Sinf[X]$,
there is a set $Y \subseteq X$ and a $Y$-separating monomial~$m$
such that the parameter~$e_Y(m)$ is finite.
\end{restatable}

\begin{proof}
	Clearly, there is a monomial~$m$ in either~$p$ or~$q$ that is
	not absorbed by any monomial from the other polynomial,
	otherwise, $p$ and $q$ would absorb each other,
	which would imply $p = q$.
	Pick $Y \coloneqq \{x \in X \mid m(x) \neq \infty\}$.
	It follows that $e_Y(m)$ is finite and
	that~$m$ is not $Y$-absorbed by any monomial from the other polynomial,
	since any~$m'$ that $Y$-absorbs~$m$ would also absorb~$m$ entirely.
\end{proof}

For example, $x^n y^\infty$ and $x^\infty y^\infty$
are $\{x\}$-separated by $m \coloneqq x^n y^\infty$ with $e_{\{x\}}(x^n y^\infty) = n$.
Now, soundness of~$G$ on~$\Sinf[X]$ follows by proving that the parameter from \cref{lem-soundSinfG}
can be exploited to limit the number of turns
required by Spoiler to win $G(\pi_A, \pi_B)$ on separable $\Sinf[X]$-interpretations
in a similar fashion as already described in \cref{thm-soundNinfG} for~$\N^{\infty}$.
	
\begin{restatable}{theorem}{THMsoundSinfG} \label{thm-soundSinfG}
	Fix some $k \ge 1$.
	Let $\pi_A$ and $\pi_B$ be $\Sinf[X]$-interpretations with
	elements $\bar{a} = (a_1, \dots, a_n)$ and $\bar{b} = (b_1, \dots, b_n)$.
	If there is a separating formula $\varphi(\bar{x})$ with $\qr (\varphi(\bar{x})) \leq m$,
	a set $Y \subseteq X$ and a separating monomial~$m$
	for $\pi_A \llb \varphi(\bar{a}) \rrb$ and $\pi_B \llb \varphi(\bar{b}) \rrb$
	such that $e_Y(m) < k$, then Spoiler wins $G_{k m}(\pi_A, \bar{a}, \pi_B, \bar{b})$.
\end{restatable}

\begin{proof}
	We show the claim by structural induction
	on the separating formula $\varphi(\bar{x})$.
	Since~$\pi_A$ and~$\pi_B$ are interchangeable,
	we may assume w.l.o.g.\@ that the $Y$-separating monomial~$m$
	is part of~$\pi_A \llb \varphi(\bar{a}) \rrb$.
	If $\varphi(\bar{x})$ is a literal, Spoiler wins immediately.
	
	\begin{itemize}
		\item
		If $\varphi (\bar{x}) = \varphi_1 (\bar{x}) \lor \varphi_2 (\bar{x})$,
		the $Y$-separating monomial~$m$ must be part of
		$\pi_A \llb \varphi_i(\bar{a}) \rrb$ for some $i \in \{1, 2\}$,
		but by definition, it cannot be $Y$-absorbed by any monomial in
		$\pi_B \llb \varphi_i(\bar{b}) \rrb$.
		Thus, $m$ $Y$-separates
		$\pi_A \llb \varphi_i(\bar{a}) \rrb$ from $\pi_B \llb \varphi_i(\bar{b}) \rrb$
		and the claim follows by induction.
		
		\item
		If $\varphi(\bar{x}) = \exists x \psi(\bar{x}, x)$,
		then $\pi_A \llb \varphi(\bar{a}) \rrb = \sum_{a \in A} \pi_A \llb \psi(\bar{a}, a) \rrb$,
		and analogously to the previous case,
		we observe that~$m$ is part of
		$\pi_A \llb \psi(\bar{a}, a) \rrb$ for some $a \in A$,
		but not $Y$-absorbed by any $\pi_B \llb \psi(\bar{b}, b) \rrb$ for $b \in B$,
		thus, Spoiler can pick such an element~$a \in A$ and win
		the remaining subgame by induction hypothesis.
		
		\item		
		If $\varphi (\bar{x}) = \varphi_1 (\bar{x}) \land \varphi_2 (\bar{x})$,
		the $Y$-separating monomial~$m = m_1 \cdot m_2$ is obtained by multiplying
		two monomials with $m_i \in \pi_A \llb \varphi_i(\bar{a}) \rrb$ for $i \in \{1, 2\}$.
		There is at least one $i \in \{1, 2\}$ such that $m_i$ $Y$-separates
		$\pi_A \llb \varphi_i(\bar{a}) \rrb$ from $\pi_B \llb \varphi_i(\bar{b}) \rrb$,
		otherwise, each $m_i$ would be $Y$-absorbed by
		some $m_i' \in \pi_B \llb \varphi_i(\bar{b}) \rrb$,
		which would yield a contradiction,
		since $m' = m_1' \cdot m_2' \in \pi_B \llb \varphi(\bar{b}) \rrb$
		would $Y$-absorb~$m$.
		Clearly, $e_Y(m_i) \le e_Y(m) < k$,
		hence Spoiler wins by invoking the induction hypothesis on the suitable subformula.
		
		\item		
		If $\varphi(\bar{x}) = \forall x \psi(\bar{x}, x)$,
		then $\pi_A \llb \varphi(\bar{a}) \rrb =
		\prod_{a \in A} \pi_A \llb \psi(\bar{a}, a) \rrb$.
		Decompose the monomial~$m$ into $m = \prod_{a \in A} m_a$
		such that $m_a \in \pi_A \llb \psi(\bar{a}, a) \rrb$ holds for all $a \in A$.
		It follows that $e_Y(m) = \sum_{a \in A} e_Y(m_a) < k$,
		thus~$e_Y(m_a)$ is nonzero for $\ell < k$
		elements $a_1, \dots, a_{\ell} \in A$ and zero otherwise.
		Spoiler picks those elements and Duplicator replies with $b_1, \dots, b_{\ell}$.
		\begin{itemize}
			\item
			If there is any $1 \le i \le \ell$ such that~$m_{a_i}$
			is not $Y$-absorbed by any monomial in $\pi_B \llb \psi(\bar{b}, b_i) \rrb$,
			then~$m_{a_i}$ $Y$-separates $\pi_A \llb \psi(\bar{a}, a_i) \rrb$
			from $\pi_B \llb \psi(\bar{b}, b_i) \rrb$,
			and together with $e_Y(m_{a_i}) \le e_Y(m) < k$,
			we can apply the induction hypothesis.
			\item
			Otherwise, each $m_{a_i}$ is $Y$-absorbed
			by some $m_{b_i} \in \pi_B \llb \psi(\bar{b}, b_i) \rrb$.
			Since $\prod_{i = 1}^{\ell} m_{b_i}$ $Y$-absorbs~$m$,
			it is impossible that each $\pi_B \llb \psi(\bar{b}, b) \rrb$
			for $b \in B \setminus \{b_1, \dots, b_{\ell}\}$
			contains some monomial~$m'$ with $e_Y(m') = 0$.
			Otherwise, those monomials would not contribute anything to
			the exponents of variables $x \in Y$,
			and their product together with $m_{b_1}, \dots, m_{b_{\ell}}$
			would result in a monomial $m'' \in \pi_B \llb \varphi(\bar{b}) \rrb$
			that $Y$-absorbs~$m$, contradicting the definition of~$m$.
			Now, it only remains for Spoiler to pick
			some~$b \in B \setminus \{b_1, \dots, b_{\ell}\}$
			such that $\pi_B \llb \psi(\bar{b}, b) \rrb$
			only contains monomials~$m'$ with $e_Y(m') > 0$.
			Duplicator must answer $a \in A \setminus \{a_1, \dots, a_{\ell}\}$,
			but then $e_Y(m_a) = 0$, hence $m_a$ $Y$-separates
			$\pi_A \llb \psi(\bar{a}, a) \rrb$ from $\pi_B \llb \psi(\bar{b}, b) \rrb$
			and we can apply the induction hypothesis. \qedhere
		\end{itemize}
	\end{itemize}
\end{proof}

\begin{corollary}
	The game $G$ is sound for $\equiv$ on the semirings $\WX, \N^\infty$ and $\Sinf[X]$.
\end{corollary}	
	
However, $G$ is unsound for some important semirings.
We construct a counterexample in the tropical semiring (which is isomorphic to the Viterbi semiring $\V$) and transfer it to the isomorphic variant  $\mathbb{D}$ of $\mathbb{L}$ by making sure that the valuations are in the interval $[0,1]$, and that the separating formula does not evaluate to a semiring element greater than $1$ in both interpretations.
The main idea behind the construction is that, given a sequence $(s_i)_{\geq 1}$ of edge labels, Spoiler cannot distinguish an infinite star with exactly $i$ edges labelled with $s_i \in \Trop$ for each $i \in \N$ from an infinite star where $\min(i,m)$ edges are labelled with $s_i$ (see \cref{fig:soundVG}).
However, for an appropriate sequence of edge labels such star graphs with distinguished centre nodes can be separated in $\fo$ by summing up all edge labels using the formula $\psi (x) = \forall y (x=y \vee Exy)$.

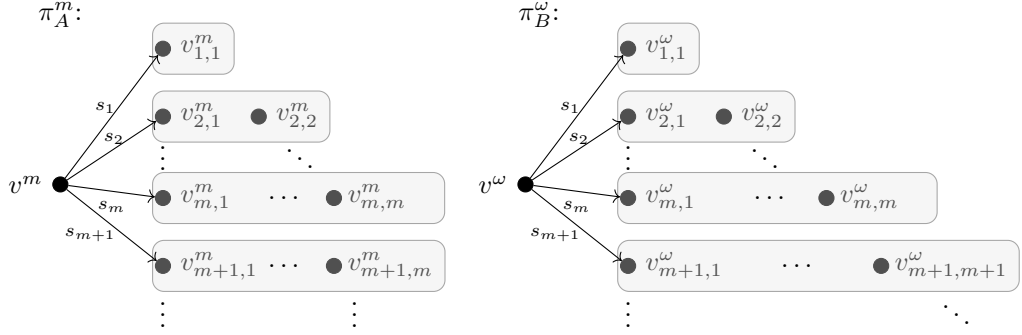
\begin{figure}[ht]
	\centering
	\begin{tikzpicture}[scale=.9]
	\tikzstyle{vertex}=[shape=circle,draw=black, fill=black, inner sep=2pt]
	\tikzstyle{box}=[draw, fill=gray!20, opacity=.35, rounded corners=1ex, inner sep=0.3mm]
	
	%piA
	\node (PIA) at (-.5,2.5) {$\pi_A^{m}$:};
	
	%vertices
	\node[vertex] (C) at (-.5,0) {};
	\node[vertex] (CA) at (1,2) {};
	\node[vertex] (CB) at (1,1) {};
	\node[vertex] (CBB) at (2.4,1) {};
	\node[vertex] (CC) at (1,-.2) {};
	\node[vertex] (CCC) at (3.5,-.2) {};
	\node[vertex] (CD) at (1,-1.2) {};
	\node[vertex] (CDD) at (3.5,-1.2) {};
	
	%vertex labels
	\node[left=0cm of C] (C') {$v^m$};
	\node[right=0cm of CA] (CA') {$v^m_{1,1}$};
	\node[right=0cm of CB] (CB') {$v^m_{2,1}$};
	\node[right=0cm of CBB] (CBB') {$v^m_{2,2}$};
	\node[right=0cm of CC] (CC') {$v^m_{m,1}$};
	\node[right=-.05cm of CCC] (CCC') {$v^m_{m,m}$ \ \ \ };
	\node[right=0cm of CD] (CD') {$v^m_{{m+1}, 1}$};
	\node[right=-.05cm of CDD] (CDD') {$v^m_{{m+1},m}$};
	
	%boxes
	\node[box, fit=(CA) (CA') (CA) (CA')](R1) {};
	\node[box, fit=(CB) (CBB') (CB) (CBB')](R2) {};
	\node[box, fit=(CC) (CCC') (CC) (CCC')](R3) {};
	\node[box, fit=(CD) (CDD') (CD) (CDD')](R4) {};
	
	%dots
	\node (vP1) at (1,.5) {$\vdots$};
	\node (vP2) at (1,-1.8) {$\vdots$};
	\node (vP3) at (3.8,-1.8) {$\vdots$};
	\node (dP1) at (3,.5) {$\ddots$};
	\node (hP1) at (2.8,-.2) {$\dots$};
	\node (hP2) at (2.8,-1.2) {$\dots$};
	
	%edges
	\path [->] (C) edge node[above] (al1) {} (CA); 
	\path [->] (C) edge node[above] (al2) {} (CB);
	\path [->] (C) edge node[below] (alm) {} (CC);
	\path [->] (C) edge node[below] (alm1) {} (CD);
	
	%edge labels
	\node[left=-.3cm of al1] (lal1) {\scriptsize $s_1$};
	\node[below right=-.5cm of al2] (lal2) {\scriptsize $s_2$};
	\node[above=-.4cm of alm] (lalm) {\scriptsize $s_m$};
	\node[left=-.3cm of alm1] (lalm1) {\scriptsize $s_{m+1}$};
	
	%______________________________________________________________________
	
	%piB
	\node (PIB) at (6.5,2.5) {$\pi_B^{\omega}$:};
	
	%vertices
	\node[vertex] (nC) at (6.3,0) {};
	\node[vertex] (nCA) at (7.8,2) {};
	\node[vertex] (nCB) at (7.8,1) {};
	\node[vertex] (nCBB) at (9.2,1) {};
	\node[vertex] (nCC) at (7.8,-.2) {};
	\node[vertex] (nCCC) at (10.7,-.2) {};
	\node[vertex] (nCD) at (7.8,-1.2) {};
	\node[vertex] (nCDD) at (11.5,-1.2) {};
	
	%vertex labels
	\node[left=0cm of nC] (nC') {$v^\omega$};
	\node[right=0cm of nCA] (nCA') {$v^\omega_{1,1}$};
	\node[right=0cm of nCB] (nCB') {$v^\omega_{2,1}$};
	\node[right=0cm of nCBB] (nCBB') {$v^\omega_{2,2}$};
	\node[right=0cm of nCC] (nCC') {$v^\omega_{m,1}$};
	\node[right=-.05cm of nCCC] (nCCC') {$v^\omega_{m,m}$ \ \ \ };
	\node[right=0cm of nCD] (nCD') {$v^\omega_{{m+1}, 1}$};
	\node[right=-.05cm of nCDD] (nCDD') {$v^\omega_{{m+1},m+1}$};
	
	%boxes
	\node[box, fit=(nCA) (nCA') (nCA) (nCA')](nR1) {};
	\node[box, fit=(nCB) (nCBB') (nCB) (nCBB')](nR2) {};
	\node[box, fit=(nCC) (nCCC') (nCC) (nCCC')](nR3) {};
	\node[box, fit=(nCD) (nCDD') (nCD) (nCDD')](nR4) {};
	
	%dots
	\node (nvP1) at (7.8,.5) {$\vdots$};
	\node (nvP2) at (7.8,-1.8) {$\vdots$};
	\node (ndP1) at (9.8,.5) {$\ddots$};
	\node (ndP2) at (12.6,-1.8) {$\ddots$};
	\node (nhP1) at (9.9,-.2) {$\dots$};
	\node (nhP2) at (10.3,-1.2) {$\dots$};
	
	%edges
	\path [->] (nC) edge node[above] (bl1) {} (nCA); 
	\path [->] (nC) edge node[above] (bl2) {} (nCB);
	\path [->] (nC) edge node[below](blm) {} (nCC);
	\path [->] (nC) edge node[below](blm1) {} (nCD);
	
	%edge labels
	\node[left=-.3cm of bl1] (lbl1) {\scriptsize $s_1$};
	\node[below right=-.5cm of bl2] (lbl2) {\scriptsize $s_2$};
	\node[above=-.4cm of blm] (lblm) {\scriptsize $s_m$};
	\node[left=-.3cm of blm1] (lblm1) {\scriptsize $s_{m+1}$};
\end{tikzpicture}
	\caption{Infinite star graphs used to construct a counterexample against the soundness of the game $G$ with respect to $\mathbb{T}$- and $\mathbb{D}$-interpretations. The grey boxes are meant to indicate $\pi^{m} \llb E v^m v^m_{i,j} \rrb = \pi^{\omega} \llb E v^m v^m_{i,j} \rrb = s_i$ for each $j$. Non-edges are assigned their Boolean truth value.}
	\label{fig:soundVG}
\end{figure}

\begin{lemma} \label{lem-soundVG}
	There is a sequence $(s_i)_{i \geq 1}$ of real numbers in $[0,1]$ such that for each $m \in \N_{>0}$
	$$
	1> \sum\limits_{i \geq 1} i \cdot s_i > \sum\limits_{i \geq 1} \min(i,m) \cdot  s_i.
	$$
\end{lemma}

\begin{proof}
	We prove the claim for $(s_i)_{i \geq 1}$ where $s_i := \frac{1}{i \cdot 2^{i+1}}$. Due to convergence of the geometrical series we obtain that $\sum_{i \in \N_{>0}} i \cdot s_i = 0.5$. Further,
	$$
	\sum\limits_{i \geq 1} i \cdot s_i = \sum\limits_{i \geq 1} \min(i,m) \cdot s_i + \underbrace{\sum\limits_{i > m } (m-i) \cdot s_i}_{>0} > \sum\limits_{i \geq 1} \min(i,m) \cdot s_i,
	$$
	which implies the claim.
\end{proof}

In order to ensure that Duplicator wins the game $G_m$ for each $m \in \N$ on single semiring interpretations $\pi$ and $\pi'$, we combine the star graphs $\pi^m$ for arbitrarily large $m$.
The idea is to include in both $\pi$ and $\pi'$ the star graphs $\pi^m$ for each $m \in \N$ as disjoint subgraphs, and to add an additional copy of $\pi^\omega$ to $\pi'$ only.
Using the sequence of edge labels satisfying $\sum_{i\leq 1} i \cdot s_i > \sum_{i\leq 1} \min (i,m) \cdot s_i$ for each $m \in \N_{>0}$ yields $\pi^\omega \llb \psi (v^\omega) \rrb > \pi^m \llb \psi (v^m) \rrb$, so the additional subgraph $\pi^\omega$ in $\pi'$ would not contribute to the valuation of the sentence $\exists x \psi (x)$.
Hence, we add additional vertices to the star graphs $\pi^m$ in both $\pi$ and $\pi'$ which increase the sum over all outgoing edges and cause $\exists x \psi (x)$ to separate the resulting semiring interpretations.

\begin{restatable}{theorem}{THMsoundVG} \label{thm-soundVG}
	The game $G$ is not sound for $\equiv$ on $\Trop, \Doubt, \Vit$ and $\Lukas$.
\end{restatable}

\begin{proof}
	Let $\Semi \in \{\Trop, \Doubt\}$ and $(s_i)_{i \geq 1}$ be defined by $s_i:=  \frac{1}{i \cdot 2^{i+1}}$.
	Further, let $s_\infty^m$ denote $\sum_{i \geq 1} \min(i,m) \cdot s_i$ for each $m \geq 1$.
	We inductively define a function $f \colon \N \setminus \{0\} \to  \N \setminus \{0\}$ which determines the number of additional nodes that are added to the star graphs. Let $f(1)$ be the smallest number such that $s_\infty^1 + f(1) \cdot s_1 > 0.5$.
	For $m>1$, we define $f(m)$ as the minimum number yielding $s_\infty^m + f(m) \cdot s_m \geq s_\infty^{m-1} + f(m-1) \cdot s_{m-1}$.
	Since $0< s_i < 1$ for all $i \geq 1$, $f$ is well-defined.
	Hence, we obtain a chain $s_\infty^1 + f(1) \cdot s_1 \leq s_\infty^2 + f(2) \cdot s_2 \leq \dots$ which is strictly upper bounded by $0.5$.
	Based on $f$ and $(s_i)_{i \geq 1}$, we construct $\Semi$-interpretations $\pi$ and $\pi'$ over the vocabulary $\tau = \{E\}$ consisting of a binary relation symbol. The universes $V$ and $V'$ are composed as follows.
	\begin{align*}
		V &= \ \{v^m \colon m \geq 1\} \cup \{v^m_{i,j} \colon j \leq \min(i,m)\} \cup \{v^m_{m,m+j} \colon j \leq f(m)\} \\
		V' &= V \ \cup \{v^\omega\} \cup \{v^\omega_{i,j} \colon j \leq i\}
	\end{align*}
	The valuations in $\pi$ and $\pi'$ are defined according to the following rules, which apply to all $m,n,i,j \in \N_{>0}$ with $m \neq n$ such that the respective nodes are contained in $V$ or $V'$.
	\begin{itemize}
		\item $\pi (E v^m v^m_{i,j})= \pi' (E v^m v^m_{i,j}) = \pi' (E v^\omega v^\omega_{i,j}) = s_i$
		\item $\pi (E v^m v^n_{i,j})= \pi' (E v^m v^n_{i,j}) = \pi' (E v^\omega v^m_{i,j}) = \pi' (E v^m v^\omega_{i,j}) = 1$
		\item $\pi (E v^m v^n)= \pi' (E v^m v^n) = \pi' (E v^\omega v^m)= \pi' (E v^m v^\omega) = 1$
	\end{itemize}
	Further, the negations of the instantiated $\tau$-literals defined above are valuated with $0$. All remaining unnegated $\tau$-literals over $V$ and $V'$ are valuated with $0$ and their negations with $1$.
	In both $\Trop$ and $\Doubt$, we obtain the following valuations of the formula $\psi(x) = \forall y (x = y \vee Exy)$.
	\begin{itemize}
		\item $\pi \llb \psi(v^m_{i,j} ) \rrb = \pi' \llb \psi(v^m_{i,j} ) \rrb = \pi' \llb \psi(v^\omega_{i,j} ) \rrb = 0$
		\item $\pi \llb \psi( v^m ) \rrb = \pi' \llb \psi( v^m ) \rrb = s_\infty^m + f(m) \cdot s_m$
		\item $\pi' \llb \psi( v^\omega ) \rrb = 0.5$
	\end{itemize}
	By construction of $f$, this implies
	$
	\pi_A \llb \exists x \psi(x) \rrb = s_\infty^1 + f(1) \cdot s_1 > 0.5 = \pi_B \llb \exists x \psi(x) \rrb,
	$
	hence $\pi_A \not\equiv_2 \pi_B$.
	In order to construct a winning strategy for Duplicator in the game $G(\pi, \pi')$, let $V^n_0 = \{v^n\}$ and $V^n_i$ for $i \geq 1$ contain all elements $v^n_{i,j}$ in $V$. We consider the partition $\mathcal{P} :=\{V^n_i \colon n \geq 1, i \geq 0\}$ of $V$ and $\mathcal{P}' := \mathcal{P} \cup \{V^\omega_i \colon i \geq 0\}$ of $V'$.
	Based on the number of turns $m$ Spoiler chooses in the game $G(\pi_A, \pi_B)$, we define a bijection $g_m \colon \mathcal{P} \to \mathcal{P}'$ as follows.
	\begin{align*}
		g_m(V^n_i) := \left\{\begin{array}{ll} V^n_i, & n<m \\
			V^\omega_i, & n=m \\
			V^{n-1}_i, & n>m \end{array}\right.
	\end{align*}
	Duplicator wins the game $G_m(\pi, \pi')$ by responding to any element in $V^n_i \subseteq V$ with an arbitrary element in $g_m(V^n_i)$ and every element in $V^n_i \subseteq V'$ with any element in $g_m^{-1} (V_i^n)$, merely making sure that (in)equalities with regard to the elements that have already been chosen are respected. This is possible, because for each $V^n_i$ we have that $|V^n_i| = |g_m(V^n_i)|$ or that $|V^n_i| \geq m$ and $|g_m(V^n_i)| \geq m$.
\end{proof}
	
\subsection{Completeness and incompleteness of the game \texorpdfstring{$G$}{G}}
	
We now turn to the study of completeness.
Analogously to $m$-turn Ehrenfeucht--Fraïssé games,
the game $G$ cannot be complete for semirings where
elementary equivalence and isomorphism of finite interpretations do not coincide,
since Duplicator clearly loses~$G$ on non-isomorphic finite interpretations.
In the remaining cases, $G$ must be complete with respect to finite interpretations,
because Spoiler winning the game implies non-isomorphism,
but on finite interpretations, this already implies separability by a first-order formula.
	
	\begin{proposition}
		Let $\Semi \in \{\Trop, \Vit, \N, \N [X]\}$.
		If Spoiler wins $G(\pi_A, \pi_B)$ and $\pi_A$, $\pi_B$ are finite
		$\Semi$-interpretations, then $\pi_A \not\equiv \pi_B$.
	\end{proposition}
	
The question arises whether this completeness result  can be lifted to infinite semiring interpretations.
For the tropical semiring $\Trop$ we  describe a counterexample
which proves that $G$ is incomplete for $\equiv$ on $\Trop$
(and hence also on~$\V$ due to $\V \cong \Trop$).

\begin{restatable}{theorem}{THMcomplVG} \label{thm-complVG}
		There are $\Trop$-interpretations $\pi_A, \pi_B$ such that Spoiler wins $G_1(\pi_A, \pi_B)$ although $\pi_A \equiv \pi_B$.
	\end{restatable}

	\begin{proof}
	Let $\pi_A$ and $\pi_B$ be $\Trop$-interpretations with just one unary predicate $R$ and  universes $A:=\{a_i :  i \in \N\}$ and $B:=\{b_i : i \in \N\}$,
	whose valuations are $\pi_A(Ra_i) = \pi_B (Rb_i) = 0$ if $i$ is even, while $\pi_A(Ra_i) = 1$ and $\pi_B(Rb_i) = 2$ for all odd $i$; since the interpretations are assumed to be model-defining this implies that $\pi_A(\lnot Ra_i) = \pi_B(\lnot Rb_i) = \infty$ for all $i \in \N$.
	Clearly, Spoiler wins $G_1(\pi_A, \pi_B)$.  To prove that $\pi_A \equiv \pi_B$, we first show that for each formula
	$\phi(\bar x)$ the valuations $\pi_A \llb \varphi(\bar{a}) \rrb$ and  $\pi_B \llb \varphi(\bar{b}) \rrb$
	can only take the values $0$ and $\infty$ if the tuples $\bar a$ and $\bar b$ only  consist of even elements
	$a_{2\ell}$ and $b_{2\ell}$. The reasoning is identical for both interpretations so we just consider $\pi_A$,
	and proceed by induction on $\phi(\bar x)$.
	For literals the claim holds by definition and for conjunctions and disjunctions it follows since $\{0,\infty\}$ is closed under the operations $\min$ and +.
	
	Consider $\varphi (\bar{x}) = \exists y \psi (\bar{x},y)$.
	For all $a \in A$ with $\pi_A (Ra)=0$, it follows by the induction hypothesis that $\pi_A \llb \psi(\bar{a},a) \rrb \in \{0,\infty\}$.
	If there is some $a \in A$ such that $\pi_A \llb \psi(\bar{a},a) \rrb=0$, it immediately follows that $\pi_A \llb \varphi (\bar{a}) \rrb = \inf_{a \in A} \pi_A \llb \psi (\bar{a},a) \rrb = 0$. Hence, it remains to show the claim for the case $\pi_A \llb \psi(\bar{a},a) \rrb = \infty$ for all $a \in A$ with $\pi_A (Ra)=0$.
	Fix some $c \in A$ that is not contained in $\bar a$ such that  $\pi_A (R c) = 0$. 
	For each $a \in A$ with $\pi_A (R a) = 1$ it holds, by monotonicity of the semiring operations,  that $\pi_A \llb \psi(\bar{a},a) \rrb \geq \pi_A \llb \psi(\bar{a},c) \rrb $ with respect to the usual order on $\R^\infty_+$ (which is the inverse of the natural order on $\Trop$)
	and since $\pi_A \llb \psi(\bar{a},c) \rrb=\infty $ we have that $\pi_A \llb \varphi (\bar{a}) \rrb = \inf_{a \in A} \pi_A \llb \psi (\bar{a},a) \rrb =\infty$.
	
	Finally, let $\varphi(\bar{x}) = \forall y \psi (\bar{x}, y)$. 
	Again, for all $a \in A$ with $\pi_A (Ra)=0$ it holds that $\pi_A \llb \psi(\bar{a},a) \rrb \in \{0,\infty\}$ by induction hypothesis.
	If there is an $a \in A$ such that $\pi_A \llb \psi(\bar{a},a) \rrb=\infty$, it immediately follows that $\pi_A \llb \varphi (\bar{a}) \rrb = \sum_{a \in A} \pi_A \llb \psi (\bar{a},a) \rrb = \infty$.
	Therefore it remains to show the claim for the case that $\pi_A \llb \psi(\bar{a},a) \rrb= 0$ for all $a \in A$ with $\pi_A (Ra)=0$.
	We observe that for all $a,a' \in A$ that do not occur in $\bar a$ with $\pi_A (Ra) = \pi_A (Ra')$ it holds that $(\pi_A , \bar{a}, a) \cong (\pi_A, \bar{a}, a')$.
	Hence, if there was some $a \in A$ with $\pi_A (Ra) = 1$ such that $\pi_A \llb \psi (\bar{a},a) \rrb = s$ for some $s>0$, then $\pi_A \llb \psi (\bar{a},a) \rrb =s$ would hold for all $a \in A$ with $\pi_A (Ra) = 1$, which implies $\pi_A \llb \varphi (\bar{a}) \rrb = \sum_{a \in A} \pi_A \llb \psi (\bar{a},a) \rrb = \infty$.
	Otherwise, we have that $\pi_A \llb \psi (\bar{a},a) \rrb = 0$ for all $a \in A$, thus $\pi_A \llb \varphi (\bar{a}) \rrb = 0$, which completes the induction.
	
	\medskip
	In particular we have for every sentence  $\varphi \in \fo (\{R\})$ that $\pi_A \llb \varphi \rrb, \pi_B \llb \varphi \rrb \in \{0,\infty\}$.
	We claim that $\pi_A \llb \varphi \rrb =\pi_B \llb \varphi \rrb$. 
	The function $h \colon \Trop \to \Trop$ defined by $s \mapsto 2s$ 
	is an endomorphism on $\Trop$ that is compatible with the infinitary operations, and obviously, $(h \circ \pi_A) \cong \pi_B$.
	If $\pi_A \llb \phi \rrb = 0$, then $\pi_B \llb \phi \rrb = 2 \cdot 0=0$ due to the fundamental property.
	Otherwise, $\pi_A \llb \phi \rrb = \infty = 2 \cdot \infty = \pi_B \llb \phi \rrb$, hence $\pi_A \equiv \pi_B$.
\end{proof}

The natural semiring does not admit infinitary operations, so we consider its extension~$\N^\infty$ instead. But on~$\N^{\infty}$,
counterexamples disproving completeness also exist.  
Despite the completeness of $m$-turn bijection games on the natural semiring $\N$, one can construct elementarily equivalent semiring interpretations 
with infinite universes on the extended semiring $\Ninf$, on
which Spoiler even wins the game $G_1$.  To prove this, we make use of the truncated  semiring $\N_{\leq 2}$ which
only contains the elements $\{0,1,2\}$ and sum and product of elements $s,t$ are given by $\min(s + t, 2)$ and $\min(s \cdot t, 2)$. Using the fact that the mapping $h \colon \N^\infty \to \N_{\leq 2}$ defined by $h \colon s \mapsto \min(s, 2)$ is a homomorphism, we can derive a method for proving elementary equivalence of infinite $\N^\infty$-interpretations as follows.
	
\begin{proposition} \label{prop-complNinf}
	Let $\pi_A$, $\pi_B$ be $\N^\infty$-interpretations and $h \colon \N^\infty \to \N_{\leq 2}$ with $n \mapsto \min(n,2)$. If the universes of $\pi_A$ and $\pi_B$ can be partitioned into infinite sets $\{A_i \colon i \in I\}$ and $\{B_i \colon i \in I\}$ such that for each $i \in I$ and  all $a,a' \in A_i$ and $b,b' \in B_i$,
	\begin{align*}
		(h \circ \pi_A, a) \equiv (h \circ \pi_A, a') \equiv (h \circ \pi_B, b) \equiv (h \circ \pi_B, b')
	\end{align*}
	then $\pi_A \equiv \pi_B$.
\end{proposition}

\begin{proof}
	Let $\pi_A$, $\pi_B$ and partitions $\{A_i \colon i \in I\}$ and $\{B_i \colon i \in I\}$ be given as above. Applying the fundamental property yields for each $i \in I$ that
	\begin{bracketenumerate}
		\item $\pi_A \llb \psi(a) \rrb = 0$ $\Leftrightarrow$ $\pi_A \llb \psi(a') \rrb = 0$ $\Leftrightarrow$ $\pi_B \llb \psi(b) \rrb = 0$ $\Leftrightarrow$ $\pi_B \llb \psi(b') \rrb = 0$,
		\item $\pi_A \llb \psi(a) \rrb = 1$ $\Leftrightarrow$ $\pi_A \llb \psi(a') \rrb = 1$ $\Leftrightarrow$ $\pi_B \llb \psi(b) \rrb = 1$ $\Leftrightarrow$ $\pi_B \llb \psi(b') \rrb = 1$ and
		\item $\pi_A \llb \psi(a) \rrb \geq 2$ $\Leftrightarrow$ $\pi_A \llb \psi(a') \rrb \geq 2$ $\Leftrightarrow$ $\pi_B \llb \psi(b) \rrb \geq 2$ $\Leftrightarrow$ $\pi_B \llb \psi(b') \rrb \geq 2$
	\end{bracketenumerate}
	for all formulae $\psi(x) \in \fo(\tau)$ and all $a,a' \in A_i$ and $b,b' \in B_i$. This implies for each $i \in I$ and all formulae $\psi(x)$ that
	\begin{align*}
		\sum\limits_{a \in A_i} \pi_A \llb \psi(a) \rrb = \sum\limits_{b \in B_i} \pi_B \llb \psi(b) \rrb \ \text{ and } \ \prod\limits_{a \in A_i} \pi_A \llb \psi(a) \rrb = \prod\limits_{b \in B_i} \pi_B \llb \psi(b) \rrb,
	\end{align*}
	since each $A_i$ and $B_i$ is infinite. Thus, we have that $\pi_A \llb \varphi \rrb = \pi_B \llb \varphi \rrb$ for all sentences $\varphi = Q x \psi (x)$ with $Q \in \{\exists, \forall\}$. If $\pi_A$ and $\pi_B$ were not elementarily equivalent, they would be separable by a sentence of this form. Hence, it must hold that $\pi_A \equiv \pi_B$.
\end{proof}

With \cref{prop-complNinf} it is straightforward to construct elementarily equivalent $\N^\infty$-interpretations on which Spoiler wins the game $G$, or even $G_1$.
For instance, we can fix arbitrary infinite $\N^\infty$-interpretations $\pi_A$ and $\pi_B$ of vocabulary $\tau = \{R\}$ where $R$ is a unary relation symbol such that $\pi_A (Ra) \geq 2$, $\pi_B (Rb) \geq 2$ and $\pi_A (\lnot Ra) = \pi_B (\lnot Rb) = 0$ for all $a \in A$ and $b \in B$. This ensures $(h \circ \pi_A, a) \cong (h \circ \pi_A, a') \cong (h \circ \pi_B, b) \cong (h \circ \pi_B, b')$
for all $a,a' \in A$ and $b,b' \in B$.
By the isomorphism lemma, we can apply \cref{prop-complNinf} without partitioning $A$ and $B$ into smaller sets. We obtain that $\pi_A \equiv \pi_B$, regardless of whether $\pi_A$ and $\pi_B$ share even a single valuation with regard to $R$.

\begin{corollary}\label{cor-complNinf}
	There are $\mathbb{N}^\infty$-interpretations $\pi_A$ and $\pi_B$ such that Spoiler wins $G_1(\pi_A, \pi_B)$ although $\pi_A \equiv \pi_B$.
\end{corollary}
	
Consequently, completeness of $G$ for $\equiv$ also fails on any semiring which extends $\N[X]$ and admits infinitary operations if it contains $\N^\infty$ as a subsemiring.

\section{The homomorphism game}
\label{sect:homgame}

Finally, we propose a new kind of model comparison games referred to as \emph{homomorphism games}. 
The idea is to reduce a given pair of $\Semi$-interpretations to $\Bool$-interpretations via homomorphisms.
In general, the resulting $\Bool$-interpretations are no longer model-defining, which is why their $m$-equivalence is not captured by $G_m$.
While soundness of $G_m$ for $\equiv_m$ on fully idempotent semirings $\Semi$ does not rely on the assumption that the $\Semi$-interpretations are model-defining, completeness for $\equiv_m$ even fails on $\Bool$, because a priori there is no connection between literals and their negations (see \cref{sec:1equivIncomplBint} for a counterexample). 
Thus, we consider a one-sided variant of the Ehrenfeucht--Fraïssé game, which yields a characterisation of $m$-equivalence for $\Bool$-interpretations without requiring them to be model-defining.

\subsection{One-sided games and separating homomorphism sets}

Consider two $\Semi$-interpretations $\pi_A$, $\pi_B$ and
let~$\Semi$ be naturally ordered by~$\le$.
We say that $(\pi_A, \bar{a}) \leq (\pi_B, \bar{b})$ if for every literal $L(\bar{x})$ we have $\pi_A (L(\bar{a})) \leq \pi_B (L(\bar{b}))$.
Further, we say that $(\pi_A, \bar{a}) \preceq_m (\pi_B, \bar{b})$ if $\pi_A \llb \varphi (\bar{a}) \rrb \leq \pi_B \llb \varphi (\bar{b}) \rrb$ holds for any formula $\phi (\bar{x})$ of quantifier rank at most $m$.

\begin{definition}
The one-sided game $G^\leq_m (\pi_A, \pi_B)$ is played
in the same way as $G_m (\pi_A, \pi_B)$,
but the winning condition for Duplicator,
assuming that the tuples $\bar{a}, \bar{b}$ were chosen after $m$ moves,
is extenuated to $(\pi_A, \bar{a}) \leq (\pi_B, \bar{b})$
instead of $(\pi_A, \bar{a}) \equiv_0 (\pi_B, \bar{b})$.
\end{definition}

Using monotonicity of both semiring operations with respect to the natural order, we obtain the following soundness result, which can be proved analogously to \cref{thm-soundIdem}.

\begin{proposition}
	Let $\Semi$ be any fully idempotent semiring. Then $G^\leq_m$ is sound for $\preceq_m$ on $\Semi$.
\end{proposition}  

On $\Bool$, the one-sided game $G^\leq_m$ is also complete for $\preceq_m$ even for $\Bool$-interpretations that are not model-defining.
To prove this, we inductively construct characteristic formulae $\chi^m_{\pi_A, \bar{a}} (\bar{x})$ analogous to the classical Ehrenfeucht--Fraïssé theorem, but we omit literals $\lnot R \bar{x}$ in $\chi^0_{\pi_A, \bar{a}} (\bar{x})$ if $\pi_A (R \bar{a}) = 0$. Let $\varphi^=_{\bar{a}}(\bar{x})$ define the equalities and inequalities of the elements in $\bar{a}$.
\begin{alignat*}{1}
	\chi^0_{\pi_A, \bar{a}} (\bar{x}) &\coloneqq  \varphi^=_{\bar{a}}(\bar{x}) \wedge \bigwedge \{ L(\bar{x}) \in \operatorname{Lit}_n(\tau) : \pi_A (L(\bar{a})) = 1\}
	\\
	\chi^{m+1}_{\pi_A, \bar{a}} (\bar{x}) &\coloneqq \bigwedge\limits_{a \in A} \exists x \ \chi^m_{\pi_A, \bar{a},a} (\bar{x}, x) \wedge \forall x \ \bigvee\limits_{a \in A} \chi^m_{\pi_A, \bar{a}, a} (\bar{x}, x)
\end{alignat*}

\begin{theorem} \label{thm-modEF}
	For any two $\Bool$-interpretations $\pi_A$ and $\pi_B$ with elements $\bar{a} \in A^n$ and $\bar{b} \in B^n$ and any $m \in \N$, the following are equivalent:
	
	\begin{bracketenumerate}
		\item Duplicator wins $G^{\leq}_m(\pi_A, \bar{a}, \pi_B, \bar{b})$;
		
		\item $\pi_B \llb \chi^m_{\pi_A, \bar{a}} (\bar{b}) \rrb = 1$;
		
		\item $(\pi_A, \bar{a}) \preceq_m (\pi_B, \bar{b})$.
	\end{bracketenumerate}
\end{theorem}

To derive homomorphism games from one-sided games $G_m^\leq$ on $\Bool$-interpretations, we make use of separating sets of homomorphisms,
which were introduced in \cite{GraedelMrk21}.

\begin{definition}
	Given semirings $\Semi$ and $\Semi'$, a set $H$ of homomorphisms from $\Semi$ to $\Semi'$ is called \emph{separating} if for all $s, t \in S$ with $s \neq t$ there is some $h \in H$ with $h(s) \neq h(t)$.
\end{definition}

For two given $\Semi$-interpretations $\pi_A$ and $\pi_B$ which are separable by some sentence $\psi$, we can think of the valuations $s \neq t$ of $\psi$ in $\pi_A$ and $\pi_B$, respectively, as witnesses for the separability of $\pi_A$ and $\pi_B$.
Further, whenever there is a homomorphism $h$ such that $h(s) \neq h(t)$ and $(h \circ \pi_A) \equiv_m (h \circ \pi_B)$, we can exclude the pair $(s,t)$ as a candidate for witnessing $\pi_A \not\equiv_m \pi_B$ due to the fundamental property.
Thus, separating sets of homomorphisms yield the following reduction technique.

\begin{lemma} \label{lem-sepHom}
	Let $\Semi$ and $\Semi'$ be semirings and $H$ a separating set of homomorphisms from $\Semi$ to $\Semi'$. Moreover let $\pi_A, \pi_B$ be $\Semi$-interpretations, $\bar{a} \in A^n$ and $\bar{b} \in B^n$. It holds that $(h \circ \pi_A, \bar{a}) \equiv_m (h \circ \pi_B, \bar{b})$ for all $h \in H$ if, and only if, $(\pi_A, \bar{a}) \equiv_m (\pi_B, \bar{b})$.
\end{lemma}

Based on a separating set $H$ of homomorphisms $h \colon \Semi \to \Bool$, the \emph{homomorphism game} $\HG_m (H, \pi_A, \pi_B)$ can be defined as follows. Spoiler first chooses some $h \in H$ and puts either $\pi_0 = h \circ \pi_A$ and  $\pi_1 = h \circ \pi_B$, or the other way around, i.e. $\pi_0 = h \circ \pi_B$ and  $\pi_1 = h \circ \pi_A$. Then the game $G^\leq_m (\pi_0, \pi_1)$ is played.
Using the fact that $G^\leq_m$ is sound and complete for~$\preceq_m$ even on $\Bool$-interpretations which are not model-defining, soundness and completeness of $HG_m$ for $\equiv_m$ can be stated as follows.

\begin{theorem}
	Let $\Semi$ be a semiring with a separating set $H$ of homomorphisms into $\Bool$. Given $\Semi$-interpretations $\pi_A, \pi_B$ and $\bar{a} \in A^n, \bar{b} \in B^n$, the following are equivalent for $m \in \N$:
	\begin{bracketenumerate}
		\item Duplicator wins $\HG_m(H, \pi_A, \bar{a}, \pi_B, \bar{b})$;
		\item $h(\pi_B \llb \chi^m_{h \circ \pi_A, \bar{a}} (\bar{b}) \rrb) = h(\pi_A \llb \chi^m_{h \circ \pi_B, \bar{b}} (\bar{a}) \rrb) = 1$ for each $h \in H$;
		\item $(\pi_A, \bar{a}) \equiv_m (\pi_B, \bar{b})$.
	\end{bracketenumerate}
\end{theorem}
	
\subsection{Homomorphisms from lattice semirings}

Motivated by Birkhoff's and Stone's representation theorems \cite{Birkhoff67, Stone38}, we present two explicit constructions of a separating set of homomorphisms from lattice semirings (i.e. fully idempotent and absorptive semiring) into $\Bool$, the first of which applies to finite lattice semirings only, and embed the sets into the rules of the homomorphism game.
Indeed, every semiring for which there is a separating set of homomorphisms to $\Bool$ must be a lattice semiring, since for every homomorphism $h \colon \Semi \to \Bool$ and $s,t \in S$, we have $h(s \cdot s)=h(s) \wedge h(s)=h(s)$ and $h(s + st) = h(s) \vee (h(s) \wedge k(t)) = h(s)$.
Due to absorption, we assume that the infinitary operations of a lattice semiring are given by $\sum_{i \in I} s_i := \sup \{\sum_{i \in I'} s_i | I' \subseteq I \text{ finite}\}$ and $\prod_{i \in I} s_i := \inf \{\prod_{i \in I'} s_i | I' \subseteq I \text{ finite}\}$.

\subparagraph*{Finite lattice semirings}
We construct a separating set of homomorphisms $h_s \colon \Semi \to \Bool$ which depend on a certain semiring element $s \in \Semi$.
In order to ensure the compatibility of $h_s$ with addition in $\Semi$, the element $s$ must be indecomposable with respect to addition in the following sense.

\begin{definition}
	Let $\Semi$ be a finite lattice semiring. A non-zero element $s \in \Semi$ is said to be \emph{$+$-indecomposable} if for all $r,t \in \Semi$ with $r \neq s$ and  $t \neq s$ it holds that $r+t\neq s$. We denote the set of non-zero $+$-indecomposable of elements in $\Semi$ as $idc(\Semi)$.
\end{definition}

In a min-max semiring, for instance, every non-zero element is $+$-indecomposable. By contrast, the $+$-indecomposable elements in $\operatorname{PosBool}[X]$ correspond to the monomials.

\begin{lemma}
	For each $s \in idc(\Semi)$ the mapping $h_s \colon \Semi \to \Bool$ defined by
	\begin{align*}
		h_s(t)=\left\{\begin{array}{ll} 1, &t+s = t \\
			0, & \text{otherwise}\end{array}\right.
	\end{align*}
	is a homomorphism from $\Semi$ into $\Bool$.
\end{lemma}

\begin{proof}
	Let $s \in idc(\Semi)$ be non-zero and $+$-indecomposable. 
	\begin{bracketenumerate}
		\item Since $0+s=s \neq 0$, it holds that $h_s(0) = 0$. Further, we have that $1 + s = 1 + 1 \cdot s = 1$ due to absorption, hence $h_s(1) = 1$.
		\item In order to prove that $h_s(r + t) = h_s(r) + h_s(t)$ for all $r, t \in \Semi$, it remains to show that $s + (r+t) = r + t$ is equivalent to $s + r = r$ or $s +t = t$.
		If $s + (r+t) = r + t$, then with absorption and distributivity $s r + s t = s(r+t) = s(s+r+t) =s + s(r + t) = s$. Since $s$ is $+$-indecomposable by assumption, this implies $s r = s$ or $s t=s$.
		Suppose w.l.o.g. that $s r = s$ which yields $r = r + s r = r + s$.
		For the converse implication, assume that $r +s = r$ or $t + s = t$. Clearly, both implications immediately yield $s + (r+t) = r + t$.
		\item To derive $h_s(r\cdot t) = h_s(r) \cdot h_s(t)$, we show that $s + r t = r t$ is equivalent to $s + r= r$ and $s + t = t$. If $s + r t = r t$, we can infer that $s+ r = s+(r+ rt) = (s+ rt)+r= rt + r = r$ and an analogous result for $t$.
		Conversely, suppose that $s +r = r$ and $s + t= t$. Then, $rt = (s+r) (s+ t) = s + (r \cdot t)$ follows by distributivity.
		\item Pertaining to the compatibility of $h_s$ with infinitary operations in $\Semi$, note that any infinite sum or product can be transformed into a finite sum or product due to full idempotence and the assumption that~$\Semi$ is finite.
		Thus, the proof is already complete. \qedhere
	\end{bracketenumerate} 
\end{proof}

Although we only consider the mappings $h_s$ for $+$-indecomposable $s$ to ensure that $h_s$ is a homomorphism, any two elements in $\Semi$ can be separated by some $h_s$.

\begin{lemma}
	The set $\{h_s \colon s \in idc(\Semi)\}$ is a separating set of homomorphisms from $\Semi$ to $\Bool$.
\end{lemma}

\begin{proof}
	For $t \in \Semi$ let $S_t = \{s \in idc(\Semi) \colon s + t = t\}$. Due to idempotence, we have that $t + \sum_{s \in S_t} s = t$. Since $\Semi$ is assumed to be finite, there must be a tuple $t_1, \dots, t_n \in idc(\Semi)$ with $t_1 + \dots + t_n = t$. With idempotence, this implies $t + t_i = t$, which yields $t_i \in S_t$ for each $1\leq i \leq n$. Hence, we have that $t + \sum_{s \in S_t} s = \sum_{1 \leq i \leq n} t_i + \sum_{s \in S_t} s = \sum_{s \in S_t} s$. Overall, we obtain $t = t+ \sum_{s \in S_t} s = \sum_{s \in S_t} s$.
	
	Let $r, t \in \Semi$ with $r \neq t$. Since $r = \sum_{s \in S_{r}} s$ and $t = \sum_{s \in S_{t}} s$, it must hold that $S_{r} \neq S_{t}$. Let $s$ be a witness for the inequality and assume w.l.o.g that $s \in S_{r}$. By definition of $S_{r}$, it holds that $s + r = r$, hence $h_s(r) = 1$. By contrast, $s \not\in S_{t}$ yields $s + t \neq t$ and thus $h_s(t) = 0$.
\end{proof}

As we derived an explicit construction a separating set of homomorphisms to $\Bool$ which applies to any finite lattice semiring, we can reformulate the homomorphism game as $\HG_m^{f}(\pi_A , \pi_B)$ corresponding to $\HG_m(H_{idc} , \pi_A, \pi_B)$ for finite lattice semirings as follows.

\begin{definition}
	At the beginning of each play in $\HG_m(\pi_A, \pi_B)$, Spoiler chooses either $\pi_0=\pi_A$ and $\pi_1= \pi_B$ or vice versa, and some $s \in idc(\Semi)$.
	In the $i$-th of $m$ rounds, Spoiler chooses some $a_i \in A$ or $b_i \in B$ and Duplicator has to respond with an element $a_i$ or $b_i$ in the other structure.
	Duplicator wins the play if for the chosen tuples $\bar{a}, \bar{b}$ and each $L(\bar{x}) \in \operatorname{Lit}_m(\tau)$
	$\pi_0 (L(\bar{a}))+s=\pi_0 (L(\bar{a}))$ implies $\pi_1 (L(\bar{b}))+s=\pi_1 (L(\bar{b}))$.
\end{definition}

The direct construction of the separating set of homomorphisms also allows an explicit formulation of characteristic formulae $\chi^{m,s}_{h \circ \pi_A, \bar{a}} (\bar{x})$ for each $s \in idc(\Semi)$ corresponding to the $\Bool$-interpretations $h_s \circ \pi_A$. Again $\varphi_{\bar{a}}^= (\bar{x})$ characterises the equalities and inequalities of the elements in $\bar{a}$.
\begin{align*}
	\chi^{0,s}_{\pi_A, \bar{a}} (x_1, \dots, x_n) &\coloneqq  \varphi_{\bar{a}}^= (\bar{x}) \wedge \bigwedge \{ L(\bar{x}) \in \operatorname{Lit}_n(\tau) \mid \pi_A ( \bar{a}) + s = \pi_A (L(\bar{a}))\} \\
	\chi^{m+1,s}_{\pi_A, \bar{a}} (x_1, \dots, x_n) &\coloneqq  \bigwedge\limits_{a \in A} \exists x \ \chi^{m,s}_{\pi_A, \bar{a},a} (\bar{x}, x) \wedge \forall x \ \bigvee\limits_{a \in A} \chi^{m,s}_{\pi_A, \bar{a}, a} (\bar{x}, x)
\end{align*}

In terms of the set $H_{idc} = \{h_s \colon s \in idc(\Semi)\}$, the correctness of the game $\HG^{f}_m$ for finite lattice semirings can be stated as follows.
\begin{theorem}
	The game $\HG^{f}_m$ is sound and complete for $\equiv_m$ on every finite lattice semiring~$\Semi$. More precisely, given any $\Semi$-interpretations $\pi_A, \pi_B$ and $\bar{a} \in A^n, \bar{b} \in B^n$ the following are equivalent for each $m \in \N$:
	\begin{bracketenumerate}
		\item Duplicator wins $\HG^{f}_m(\pi_A, \bar{a}, \pi_B, \bar{b})$;
		\item For each $s \in idc(\Semi)$, it holds that
		\[
		\pi_B \llb \chi^{m,s}_{\pi_A, \bar{a}} (\bar{b}) \rrb + s= \pi_B \llb \chi^{m,s}_{\pi_A, \bar{a}} (\bar{b}) \rrb
		\quad\text{and}\quad
		\pi_A \llb \chi^{m,s}_{\pi_B, \bar{b}} (\bar{a}) \rrb + s = \pi_A \llb \chi^{m,s}_{\pi_B, \bar{b}} (\bar{a}) \rrb;
		\]
		\item $(\pi_A, \bar{a}) \equiv_m (\pi_B, \bar{b})$.
	\end{bracketenumerate} 
\end{theorem}

\subparagraph*{Infinite lattice semirings} In the case of min-max semirings, the construction of the separating set of homomorphisms also applies to infinite semirings.
However, it can be shown that the constructed set $H_{idc}$ does not suffice to separate infinite lattice semirings in general.
As an example, consider the lattice semiring $\Semi = (\mathbb{Z}, +^\Semi, \cdot^\Semi, 0, 1)$ with $s+^\Semi t = \operatorname{gcd}(s,t)$ if $s \neq 0$ or $t \neq 0$, while $0+^\Semi 0 = 0$ and $s \cdot^\Semi t = \operatorname{lcm}(s,t)$ for $s,t \in \mathbb{Z}$. 
For each $s \in \mathbb{Z}$, it holds that $\operatorname{gcd}(2s,3s) =s$, so for $s \neq 0$ there are distinct $r$ and $t$ such that $s = r+^\Semi t$.
By contrast, $\operatorname{gcd} (s,t) \neq 0$ for all $s,t \in \mathbb{Z} \setminus \{0\}$, hence $idc(\Semi) = \{0\}$, but $\{h_0\}$ is not a separating set of homomorphisms.
However, a separating set of homomorphisms still exists in the infinite case, which relies on prime ideals in $\Semi$ instead of $+$-indecomposable elements.

\begin{definition}
	Let $\Semi$ be a lattice semiring. A non-empty proper subset $P$ of $S$ is said to be a \emph{prime ideal} if
	\begin{bracketenumerate}
		\item $s \in P$ and $t \in P$ imply $s + t \in P$,
		\item $s \in P$ and $t \in S$ imply $s \cdot t \in P$ and
		\item $s \cdot t \in P$ implies $s \in P$ or $t \in P$.
	\end{bracketenumerate}
	We denote the set of prime ideals in $\Semi$ by $I_p(\Semi)$.
\end{definition}

Since one can find for every pair of distinct elements $s,t\in \Semi$ a prime ideal which contains one of $s$ and $t$ but not both \cite{Stone38}, the prime ideals in $\Semi$ allow us to construct a separating set of homomorphisms $h \colon \Semi \to \Bool$.

\begin{lemma}[{\cite[Theorem 13]{Stone38}}] \label{lem-stone}
	The mapping $f \colon \Semi \to \mathcal{P}(I_p(\Semi))$, $s \mapsto \{P \in I_p (\Semi) \colon s \not\in P\}$ is injective and it holds for each $s,t \in S$ that
	\begin{bracketenumerate}
		\item $f(s + t) = f(s) \cup f(t)$ and
		\item $f(s \cdot t) = f(s) \cap f(t)$.
	\end{bracketenumerate}
\end{lemma}

\begin{corollary} \label{prop-sepHomInfLat}
	The set $H_p:=\{h_P \colon P \in I_P(\Semi)\}$ of mappings $h_P \colon \Semi \to \Bool$ with $h_P \colon s \mapsto 0$ if, and only if, $s \in P$ is a separating set of homomorphisms.
\end{corollary}

From the separating set $H_p$ of homomorphisms, we derive the following formulation of the homomorphism game, which corresponds to $\HG_m (H_p, \pi_A, \pi_B)$.

\begin{definition}
	In each play of $\HG_m^{\infty}(\pi_A, \pi_B)$, Spoiler first chooses a prime ideal $P \in I_p(\Semi)$ and puts either $\pi_0=\pi_A$ and $\pi_1=\pi_B$, or $\pi_0=\pi_B$ and $\pi_1=\pi_A$. Afterwards, Spoiler chooses some $a \in A$ or $b \in B$ and Duplicator has to respond with an element $a$ or $b$ in the other interpretation, which is repeated $m$ times. Duplicator wins the play where the tuples $\bar{a}, \bar{b}$ have been chosen if $\pi_1 (L(\bar{b})) \in P$ implies $\pi_0 (L(\bar{a})) \in P$ for each $L(\bar{x}) \in \lit_m(\tau)$.
\end{definition}
With each prime ideal $P \in I_p(\Semi)$, we associate characteristic formulae $\chi_{\pi_A, \bar{a}}^{m,P} (\bar{x})$ according to
\begin{align*}
	\chi^{0,P}_{\pi_A, \bar{a}} (\bar{x}) &:=  \varphi_{\bar{a}}^= (\bar{x}) \wedge \bigwedge \{ L(\bar{x}) \in \lit_n(\tau) : \pi_A (L (\bar{a})) \not\in P\} \text{ and} \\
	\chi^{m+1,P}_{\pi_A, \bar{a}} (\bar{x}) &:=  \bigwedge\limits_{a \in A} \exists x \ \chi^{m,P}_{\pi_A, \bar{a},a} (\bar{x}, x) \wedge \forall x \ \bigvee\limits_{a \in A} \chi^{m,P}_{\pi_A, \bar{a}, a} (\bar{x}, x),
\end{align*}
which characterise $m$-equivalence of $\Semi$-interpretations in lattice semirings as follows.

\begin{theorem}
	The game $\HG^{\infty}_m$ is sound and complete for $\equiv_m$ on every lattice semiring. More precisely, for any two $\Semi$-interpretations $\pi_A$, $\pi_B$, elements $\bar{a} \in A^n$, $\bar{b} \in B^n$ and $m \in \N$, the following are equivalent:
	\begin{bracketenumerate}
		\item Duplicator wins $\HG_m^{\infty}(\pi_A, \bar{a}, \pi_B, \bar{b})$;
		\item For each $P \in I_P(\Semi)$, it holds that $\{\pi_B \llb \chi^{m,P}_{\pi_A, \bar{a}} (\bar{b}) \rrb, \pi_A \llb \chi^{m,P}_{\pi_B, \bar{b}} (\bar{a}) \rrb \} \cap P = \varnothing$;
		\item $(\pi_A, \bar{a}) \equiv_m (\pi_B, \bar{b})$.
	\end{bracketenumerate} 
\end{theorem}

\begin{example}
	We can use the homomorphism game to show that first-order logic
	with semiring semantics cannot express the following property
	on min-max-semirings with the monadic signature~$\{Q, R\}$:
	\emph{``For the majority of elements~$e$ in the universe,
	$Qe$ has a greater value than $Re$.''}
	To prove this, we use the following two $\Semi_4$-interpretations
	on the min-max-semiring $\Semi_4$ with four elements $\{0, 1, 2, 3\}$.
	
	\medskip
	
	\begin{minipage}{0.9\linewidth}
		\centering
		$\pi_{A}:\quad$
		\begin{tabular}{c | c | c | c | c |}
			$A$ & $Q$ & $R$ & $\neg Q$ & $\neg R$ \\ \hline
			$a_1$ & 1 & 3 & 0 & 0 \\
			$a_2$ & 2 & 1 & 0 & 0 \\
			$a_3$ & 3 & 2 & 0 & 0 \\
		\end{tabular}
		$\quad\quad\quad\pi_{B}:\quad$
		\begin{tabular}{c | c | c | c | c |}
			$B$ & $Q$ & $R$ & $\neg Q$ & $\neg R$ \\ \hline
			$b_1$ & 3 & 1 & 0 & 0 \\
			$b_2$ & 1 & 2 & 0 & 0 \\
			$b_3$ & 2 & 3 & 0 & 0 \\
		\end{tabular}
	\end{minipage}
	
	\medskip
	
	Clearly, $\pi_A$ has the desired property while $\pi_B$ does not.
	However, we can show with the homomorphism games $\HG^{\infty}_m(\pi_A, \pi_B)$  that $\pi_A\equiv \pi_B$.
	First, we observe that the prime ideals~$I_p(\Semi_4)$
	are precisely the three non-empty proper downward closed subsets of~$\{0, 1, 2, 3\}$.
	They induce homomorphisms $h_{\ge i} \colon \Semi_4 \to \Bool$ for $i \in \{1, 2, 3\}$
	such that $h_{\ge i}(j) = 1$ iff $j \ge i$.
	Hence, we essentially play the homomorphism game $\HG_m(H, \pi_A, \pi_B)$
	with the separating set of homomorphisms $H = \{h_{\ge 1}, h_{\ge 2}, h_{\ge 3}\}$.
	Now, it only remains to observe that applying any of these homomorphisms
	to~$\pi_A$ and~$\pi_B$ makes them isomorphic to each other,
	thus, Duplicator clearly has a winning strategy.
	This demonstrates the viability of homomorphism games
	as a proof method for inexpressibility results in semiring semantics.
\end{example}
	
\section{Conclusion}
	
We have provided a rather detailed study of soundness and completeness of Ehrenfeucht--Fraïssé games, and related model comparison games,
for proving elementary equivalence and $m$-equivalence in semiring semantics. The general picture that emerges is quite diverse. While the 
$m$-move games $G_m$ are sound and complete for $\equiv_m$ only on the Boolean semiring, the games still provide a  sound method on
fully idempotent semirings, such as min-max semirings, lattice semirings, and the provenance semirings $\operatorname{PosBool}[X]$. 
This permits to generalise certain classical results in logic, proved via  Ehrenfeucht--Fraïssé games or back-and-forth systems, from
Boolean structures to semiring interpretations in fully idempotent semirings. A particular example is the proof of a Hanf locality theorem
for such semirings in \cite{BiziereGraNaa23}.  For proving elementary equivalence, without restriction of the quantifier rank, 
Ehrenfeucht--Fraïssé games without a fixed number of moves provide a method that is, for various reasons,  sound on more semirings, 
including not only $\N$ and $\N^\infty$ but also the provenance semirings $\WW[X]$, $\Bool[X], \mathbb{S}[X]$, $\N[X]$, and $\Sinf[X]$.
While in classical semantics, a separating sentence of quantifier rank $m$ leads to a winning strategy of Spoiler in at most $m$ moves,
the situation in semirings may be more complicated, in the sense that a winning strategy of Spoiler which ``simulates'' a separating sentence may
still exist, but may require a larger number of moves than given by the quantifier rank; as a consequence the unrestricted game $G$
may still provide a sound method for proving elementary equivalence, although the $m$-move games are unsound for $\equiv_m$.

The most straightforward application of Ehrenfeucht--Fraïssé games and other model comparison games
are inexpressibility results, showing that a property~$P$ is not expressible in a logic~$L$.
Classically, this is accomplished by constructing two structures, precisely one of which satisfies the property~$P$,
and then providing a winning strategy for Duplicator in an appropriate model comparison game on the two structures.
This method only relies on the soundness of the model comparison game without requiring completeness.
Hence, our soundness results enable us to lift inexpressibility results to semiring semantics for a significant class of semirings.
Consider, for instance, a min-max-semiring~$\Semi$ modelling access levels and $\Semi$-interpretations $\pi$
that annotate every edge of a graph with a required  access level. Then there is no first-order formula $\phi(x, y)$ such that
$\pi \llb \phi(v, w) \rrb$ evaluates to the minimal access level required to go from~$v$ to~$w$.

We have also studied bijection and counting games, and we have shown in particular, that $m$-move bijection games are sound for $\equiv_m$ on
\emph{all} semirings. We remark that these games have originally been invented in the form of $k$-pebble games for logics with counting. This means that
rather than just selecting, in $m$ turns, two $m$-tuples, the games proceed by moving
a fixed number of $k$ pairs of pebbles through the two structures in an a priori unrestricted number of moves. These games capture 
equivalences for formulae that may use at most~$k$~variables which can, however, be quantified again and again.  
We have chosen here the simplified variants of $m$-move games rather than $k$-pebble games, to study the relationship
with the classical Ehrenfeucht--Fraïssé games for $\equiv_m$. However, also the definition of $k$-pebble bijection and counting games 
extends in a straightforward way from classical structures to semiring interpretations and their soundness properties 
for $k$-variable equivalences are analogous to those of the $m$-move variants for $m$-equivalence.  But clearly, the $k$-pebble variants of these games
deserve further study, and this will be part of our future work on the subject. We conjecture that by lifting the well-known CFI-construction to 
semirings one can show that there is no semiring where first-order logic, and even fixed point logic, is strong enough to express all properties
 that are decidable in PTIME.

On the other side, it has turned out that all these model comparison games are incomplete for elementary equivalence and $m$-equivalence on
most semirings, with the exceptions of $\N$ and $\N[X]$. Most of these incompleteness results rely on the construction of 
logically equivalent semiring interpretations on which, however, Spoiler wins the games in few moves. The proof of elementarily equivalence for 
such interpretations in general relies on separating sets of homomorphisms. Based on this technique, we have proposed a new kind 
of model comparison games, homomorphism games, which in fact are sound and complete for $m$-equivalence on finite and infinite
lattice semirings. This also raises the question whether it is possible to develop further  games that are sound and complete for 
more, or even all, semirings. An essential part of the homomorphism game is a one-sided version of the classical Ehrenfeucht--Fraïssé game, with a
winning condition that is based on (weak) local homomorphisms rather than local isomorphisms, and which capture the notion that
one interpretation never evaluates to strictly larger  values than the other. This game itself is interesting also in many other contexts and
will be further studied in future work.

\appendix

\section{Incompleteness of \texorpdfstring{$G_m$ and $G$}{(m-turn) Ehrenfeucht--Fraïssé games} on \texorpdfstring{$\Bool$}{B}-interpretations that are not model-defining} \label{sec:1equivIncomplBint}
	
We first prove incompleteness of $G$ for $\equiv$ on $\Bool$-interpretations that are not model-defining and derive from the counterexample incompleteness of $G_m$ for $\equiv_m$ for each $m \in \N$.

\begin{proposition}
	Let $\pi_A$ and $\pi_B$ be $\Bool$-interpretations on universes $A:= \{a_i \colon i \in \N\}$ and $B:=\{b_i \colon i \in \N\} \cup \{b_0'\}$ defined by the following tables.
	\begin{center}
		$\pi_A:$
		\begin{tabular}{c||c|c|c|c}
			$A$ & $R_1$ & $R_2$ & $\lnot R_1$ & $\lnot R_2$ \\
			\hline
			\hline
			$a_0$ &   $1$ &   $0$ &   $0$ &   $0$ \\
			\hline
			$a_1$ & $0$ & $0$ & $0$ & $0$ \\
			\hline
			$a_2$ &   $1$ &   $1$ &   $0$ &   $0$ \\
			\hline
			$a_3$ & $0$ & $0$ & $0$ & $0$ \\
			\hline
			$a_4$ &   $1$ &   $1$ &   $0$ &   $0$ \\
			\hline
			\vdots & \vdots & \vdots & \vdots & \vdots
		\end{tabular}
		\hspace*{.5cm}
		$\pi_B:$
		\begin{tabular}{c||c|c|c|c}
			$B$ & $R_1$ & $R_2$ & $\lnot R_1$ & $\lnot R_2$ \\
			\hline
			\hline
			$b_0$ &   $0$ &   $0$ &   $0$ &   $0$ \\
			\hline
			$b_0'$ &   $1$ &   $1$ &   $0$ &   $0$ \\
			\hline
			$b_1$ & $0$ & $0$ & $0$ & $0$ \\
			\hline
			$b_2$ &   $1$ &   $1$ &   $0$ &   $0$ \\
			\hline
			$b_3$ & $0$ & $0$ & $0$ & $0$ \\
			\hline
			\vdots & \vdots & \vdots & \vdots & \vdots
		\end{tabular}
	\end{center}
	It holds that $\pi_A \equiv \pi_B$ although Spoiler wins $G_1 (\pi_A, \pi_B)$.
\end{proposition}

\begin{proof}
	We show that for each formula $\varphi(x_1, \dots, x_n)$ and $i_1, \dots, i_n \in \N$, it holds that
	$
	\pi_B \llb \varphi (\bar{b}) \rrb \leq \pi_A \llb \varphi (\bar{a}) \rrb \leq \pi_B \llb \varphi (\bar{b}') \rrb,
	$
	where $\bar{a} = (a_{i_1}, \dots, a_{i_n})$, $\bar{b} = (b_{i_1}, \dots, b_{i_n})$ and $\bar{b}' \in B^n$ coincides with $\bar{b}$ up to occurrences of $b_0$ which are substituted by $b_0'$.
	We proceed by induction on the structure of $\varphi (\bar{x})$.
	By definition, the base case where $\varphi(\bar{x})$ is a literal is satisfied.
	
	If $\varphi(\bar{x}) = \psi (\bar{x}) \vee \vartheta(\bar{x})$, suppose that $\pi_A \llb \varphi (\bar{a}) \rrb = 0$. It suffices to show that $\pi_B \llb \varphi (\bar{b}) \rrb = 0$ in this case, as $\pi_A \llb \varphi (\bar{a}) \rrb \leq \pi_B \llb \varphi (\bar{b}') \rrb$ is clearly satisfied. We have $\pi_A \llb \psi (\bar{a}) \rrb = \pi_A \llb \vartheta (\bar{a}) \rrb = 0$, implying ${\pi_B \llb \psi (\bar{b}) \rrb = \pi_B \llb \vartheta (\bar{b}) \rrb = 0}$ by induction hypothesis. Hence, $\pi_B \llb \varphi (\bar{b}) \rrb = 0$ and we obtain $\pi_B \llb \varphi (\bar{b}) \rrb \leq \pi_A \llb \varphi (\bar{a}) \rrb$. Otherwise, it must hold that $\pi_A \llb \varphi (\bar{a}) \rrb = 1$, yielding $\pi_A \llb \psi (\bar{a}) \rrb = 1$ or ${\pi_A \llb \vartheta (\bar{a}) \rrb = 1}$. By induction, $\pi_B \llb \psi (\bar{b}') \rrb = 1$ or $\pi_B \llb \vartheta (\bar{b}') \rrb = 1$. Hence, $\pi_B \llb \varphi (\bar{b}') \rrb = 1$ and we obtain $\pi_A \llb \varphi (\bar{a}) \rrb \leq \pi_B \llb \varphi (\bar{b}') \rrb$, while $\pi_B \llb \varphi (\bar{b}) \rrb \leq \pi_A \llb \varphi (\bar{a}) \rrb$ follows immediately from $\pi_A \llb \varphi (\bar{a}) \rrb = 1$.
	
	For $\varphi(\bar{x}) = \exists x \psi (\bar{x},x)$ let $\pi_A \llb \varphi (\bar{a}) \rrb = 0$. Then, it must hold that $\pi_A \llb \psi (\bar{a},a) \rrb = 0$ for all $a \in A$, which implies $\pi_B \llb \psi (\bar{b},b) \rrb = 0$ for all $b \in B \setminus \{b_0'\}$ by induction hypothesis. Fix some $b \in B$ which is not contained in $\bar{b}$ such that $\pi_B (R_1b) = \pi_B(R_2b) = 0$. It holds that $(\pi_B, \bar{b}, b_0') \cong (\pi_B, \bar{b}, b)$, so applying the isomorphism lemma yields $\pi_B \llb \psi (\bar{b},b_0') \rrb = 0$. We obtain $\pi_B \llb \varphi (\bar{b}) \rrb = 0$ overall, so $\pi_B \llb \varphi (\bar{b}) \rrb \leq \pi_A \llb \varphi (\bar{a}) \rrb$. In case $\pi_A \llb \varphi (\bar{a}) \rrb = 1$, there must be some $a_i \in A$ such that $\pi_A \llb \psi(\bar{a}, a_i) \rrb = 1$. It follows from the induction hypothesis that $\pi_B \llb \psi(\bar{b}', b_i) \rrb = 1$ if $i>0$ and $\pi_B \llb \psi(\bar{b}', b_0') \rrb = 1$ in the case $i=0$. Thus, it holds that $\pi_B \llb \varphi(\bar{b}') \rrb = 1$, which yields $\pi_A \llb \varphi (\bar{a}) \rrb \leq \pi_B \llb \varphi (\bar{b}') \rrb$.
	
	We omit the cases $\varphi(\bar{x}) = \psi (\bar{x}) \wedge \vartheta(\bar{x})$ and $\varphi(\bar{x}) = \forall x \psi (\bar{x},x)$, as they are analogous to disjunctions and universal quantifications. In particular, the inequality implies that $\pi_B \llb \varphi \rrb \leq \pi_A \llb \varphi \rrb \leq \pi_B \llb \varphi \rrb$ for all sentences $\varphi$, hence we obtain $\pi_A \equiv \pi_B$.
\end{proof}

\begin{corollary}
	For every $m \in \N$ with $m > 0$, the game $G_m$ is incomplete for $\equiv_m$ on $\Bool$-interpretations that are not model-defining.
\end{corollary}
	
\begin{proof}
Let $\pi_A^{m}$ and $\pi_B^{m}$ be the subinterpretations of $\pi_A$ and $\pi_B$ induced by $\{a_i \colon 0 \leq i \leq 2m\}$ and $\{b_i \colon 1 \leq i \leq 2m\}$.
Observe that Duplicator wins both $G_m(\pi_A^{m}, \pi_A)$ and $G_m (\pi_B^{m}, \pi_B)$
and soundness of~$G_m$ holds, thus, together with $\pi_A \equiv \pi_B$,
we have $\pi_A^{m} \equiv_m \pi_B^{m}$,
but Spoiler still wins $G_m(\pi_A^m, \pi_B^m)$ in a single turn by picking~$a_0$.
\end{proof}
	

\begin{thebibliography}{10}

\bibitem{AmsterdamerDeuTan11}
Y.~Amsterdamer, D.~Deutch, and V.~Tannen.
\newblock On the limitations of provenance for queries with difference.
\newblock In {\em 3rd Workshop on the Theory and Practice of Provenance,
  TaPP'11}, 2011.
\newblock See also CoRR abs/1105.2255.

\bibitem{Birkhoff67}
G.~Birkhoff.
\newblock {\em Lattice Theory}.
\newblock American Mathematical Society, Providence, 3rd edition, 1967.

\bibitem{BiziereGraNaa23}
C.~Bizi{\`e}re, E.~Gr{\"a}del, and M.~Naaf.
\newblock Locality theorems in semiring semantics.
\newblock In {\em Proceedings of MFCS 2023}, 2023.
\newblock Full version: arXiv 2303.12627.

\bibitem{BourgauxOzaPenPre20}
C.~Bourgaux, A.~Ozaki, R.~Pe{\~{n}}aloza, and L.~Predoiu.
\newblock Provenance for the description logic {ELHr}.
\newblock In {\em Proceedings of {IJCAI} 2020}, pages 1862--1869, 2020.
\newblock \href {http://dx.doi.org/10.24963/ijcai.2020/258}
  {\path{doi:10.24963/ijcai.2020/258}}.

\bibitem{DannertGra19}
K.~Dannert and E.~Gr\"{a}del.
\newblock Provenance analysis: A perspective for description logics?
\newblock In C.~Lutz et~al., editor, {\em Description Logic, Theory
  Combination, and All That}, Lecture Notes in Computer Science Nr. 11560.
  Springer, 2019.
\newblock \href {http://dx.doi.org/10.1007/978-3-030-22102-7_12}
  {\path{doi:10.1007/978-3-030-22102-7_12}}.

\bibitem{DannertGra20}
K.~Dannert and E.~Gr\"{a}del.
\newblock Semiring provenance for guarded logics.
\newblock In {\em Hajnal Andréka and István Németi on Unity of Science: From
  Computing to Relativity Theory through Algebraic Logic}, Outstanding
  Contributions to Logic. Springer, 2020.

\bibitem{DannertGraNaaTan21}
K.~Dannert, E.~Gr{\"a}del, M.~Naaf, and V.~Tannen.
\newblock Semiring provenance for fixed-point logic.
\newblock In {\em Proceedings of CSL 2021}, 2021.

\bibitem{EbbinghausFlu99}
H.-D. Ebbinghaus and J.~Flum.
\newblock {\em Finite Model Theory}.
\newblock Springer, 2nd edition, 1999.
\newblock \href {http://dx.doi.org/10.1007/3-540-28788-4}
  {\path{doi:10.1007/3-540-28788-4}}.

\bibitem{GeertsPog10}
F.~Geerts and A.~Poggi.
\newblock On database query languages for {K-relations}.
\newblock {\em J. Applied Logic}, 8(2):173--185, 2010.

\bibitem{GeertsUngKarFunChr16}
F.~Geerts, T.~Unger, G.~Karvounarakis, I.~Fundulaki, and V.~Christophides.
\newblock Algebraic structures for capturing the provenance of {SPARQL}
  queries.
\newblock {\em J. {ACM}}, 63(1):7:1--7:63, 2016.

\bibitem{Glavic21}
B.~Glavic.
\newblock Data provenance.
\newblock {\em Foundations and Trends in Databases}, 9(3-4):209--441, 2021.
\newblock \href {http://dx.doi.org/10.1561/1900000068}
  {\path{doi:10.1561/1900000068}}.

\bibitem{GraedelHelNaaWil22}
E.~Gr{\"{a}}del, H.~Helal, M.~Naaf, and R.~Wilke.
\newblock Zero-one laws and almost sure valuations of first-order logic in
  semiring semantics.
\newblock In Christel Baier and Dana Fisman, editors, {\em {LICS} '22: 37th
  Annual {ACM/IEEE} Symposium on Logic in Computer Science, Haifa, Israel,
  August 2 - 5, 2022}, pages 41:1--41:12. {ACM}, 2022.
\newblock \href {http://dx.doi.org/10.1145/3531130.3533358}
  {\path{doi:10.1145/3531130.3533358}}.

\bibitem{GraedelLueNaa21}
E.~Gr{\"a}del, N.~L{\"u}cking, and M.~Naaf.
\newblock Semiring provenance for {B}{\"u}chi games: Strategy analysis with
  absorptive polynomials.
\newblock In {\em Proceedings 12th International Symposium on Games, Automata,
  Logics, and Formal Verification (GandALF 2021)}, volume 346 of {\em {EPTCS}},
  pages 67--82, 2021.

\bibitem{GraedelMrk21}
E.~Gr\"{a}del and L.~Mrkonji\'{c}.
\newblock Elementary equivalence versus isomorphism in semiring semantics.
\newblock In {\em 48th International Colloquium on Automata, Languages, and
  Programming (ICALP 2021)}, volume 198, pages 133:1--133:20, Dagstuhl,
  Germany, 2021.
\newblock \href {http://dx.doi.org/10.4230/LIPIcs.ICALP.2021.133}
  {\path{doi:10.4230/LIPIcs.ICALP.2021.133}}.

\bibitem{GraedelTan17}
E.~Gr{\"a}del and V.~Tannen.
\newblock Semiring provenance for first-order model checking, 2017.
\newblock \href {http://arxiv.org/abs/1712.01980} {\path{arXiv:1712.01980}}.

\bibitem{GraedelTan20}
E.~Gr{\"a}del and V.~Tannen.
\newblock Provenance analysis for logic and games.
\newblock {\em Moscow Journal of Combinatorics and Number Theory},
  9(3):203--228, 2020.
\newblock \href {http://dx.doi.org/10.2140/moscow.2020.9.203}
  {\path{doi:10.2140/moscow.2020.9.203}}.

\bibitem{GreenIveTan09}
T.~Green, Z.~Ives, and V.~Tannen.
\newblock Reconcilable differences.
\newblock In {\em Database Theory - {ICDT} 2009}, pages 212--224, 2009.

\bibitem{GreenKarTan07}
T.~Green, G.~Karvounarakis, and V.~Tannen.
\newblock Provenance semirings.
\newblock In {\em Principles of Database Systems {PODS}}, pages 31--40, 2007.

\bibitem{GreenTan17}
T.~Green and V.~Tannen.
\newblock The semiring framework for database provenance.
\newblock In {\em Proceedings of PODS}, pages 93--99, 2017.

\bibitem{Hella92}
L.~Hella.
\newblock Logical hierarchies in {PTIME}.
\newblock In {\em Proceedings of LICS 92}, pages 360--368, 1992.

\bibitem{ImmermanLan90}
N.~Immerman and E.~Lander.
\newblock Describing graphs: A first-order approach to graph canonization.
\newblock In {\em Complexity Theory Retrospective}. Springer, 1990.

\bibitem{Stone38}
M.~H. Stone.
\newblock {Topological representations of distributive lattices and Brouwerian
  logics}.
\newblock {\em Časopis pro pěstování matematiky a fysiky}, 067(1):1--25,
  1938.
\newblock \href {http://dx.doi.org/10.21136/CPMF.1938.124080}
  {\path{doi:10.21136/CPMF.1938.124080}}.

\end{thebibliography}
\end{document}